\documentclass{article}

\usepackage{makeidx}
\usepackage{latexsym}
\usepackage{amsfonts}
\usepackage{amssymb}
\usepackage{amsmath}
\usepackage{amstext}
\usepackage{amsthm}
\usepackage{mathrsfs}
\usepackage{float}
\usepackage{multirow}
\usepackage{url}
\usepackage{color}
\usepackage{mathabx}
\usepackage{tikz}
\usepackage[colorlinks=true, urlcolor=blue, linkcolor=blue, citecolor=blue]{hyperref}

\usepackage[normalem]{ulem}

\newtheorem{theorem}{Theorem}

\newtheorem{lemma}[theorem]{Lemma}

\newtheorem{corollary}[theorem]{Corollary}
\newtheorem{definition}[theorem]{Definition}

\newtheorem{proposition}[theorem]{Proposition}

\begin{document}

\title{Manipulation of social choice correspondences \\ under incomplete information}

\author{
\textbf{Raffaele Berzi} \\
{\small {Dipartimento di Scienze per l'Economia e  l'Impresa} }\\
\vspace{-6mm}\\
{\small {Universit\`{a} degli Studi di Firenze}}\\
\vspace{-6mm}\\
{\small {via delle Pandette 9, 50127, Firenze, Italy}}\\
\vspace{-6mm}\\
{\small {e-mail: raffaele.berzi@unifi.it}}\\
\vspace{-6mm}\\
{\small https://orcid.org/0009-0008-1489-6975}\\
\and 
\textbf{Daniela Bubboloni} \\
{\small {Dipartimento di Matematica e Informatica U.Dini} }\\	
\vspace{-6mm}\\
{\small {Universit\`{a} degli Studi di Firenze} }\\
\vspace{-6mm}\\
{\small {viale Morgagni 67/a, 50134 Firenze, Italy}}\\
\vspace{-6mm}\\
{\small {e-mail: daniela.bubboloni@unifi.it}}\\
\vspace{-6mm}\\
{\small https://orcid.org/0000-0002-1639-9525}\\
 \and \textbf{Michele Gori}
 \\
{\small {Dipartimento di Scienze per l'Economia e  l'Impresa} }\\
\vspace{-6mm}\\
{\small {Universit\`{a} degli Studi di Firenze} }\\
\vspace{-6mm}\\
{\small {via delle Pandette 9, 50127, Firenze, Italy}}\\
\vspace{-6mm}\\
{\small {e-mail: michele.gori@unifi.it}}\\
\vspace{-6mm}\\
{\small https://orcid.org/0000-0003-3274-041X}}

\maketitle

\begin{abstract}
\noindent We study the manipulability of social choice correspondences in situations where individuals have incomplete information about others' preferences. We propose a general concept of manipulability that depends on the extension rule used to derive preferences over sets of alternatives from preferences over alternatives, as well as on individuals' level of information. We then focus on the manipulability of social choice correspondences when the Kelly extension rule is used, and individuals are assumed to have the capability to anticipate the outcome of the collective decision. Under these assumptions, we introduce some monotonicity and sensitivity properties for social choice correspondences that combined imply manipulability.  Then we prove a result of manipulability for unanimous positional social choice correspondences, and present a detailed analysis of the manipulability properties for the Borda, the plurality, the negative plurality and the Copeland social choice correspondences. 
\end{abstract}

\vspace{4mm}


\vspace{2mm}


\vspace{2mm}


\section{Introduction}

Consider a group of individuals who must select one or more alternatives from a given set. Assume that such a selection must be based solely on individuals' preferences, expressed through rankings of alternatives. Any procedure that associates a nonempty set of alternatives with each preference profile, that is, a complete list of individual preferences, is called a social choice correspondence ({\sc scc}). A {\sc scc} is manipulable if there are situations where an individual has an incentive to misrepresent her preferences because doing so makes the {\sc scc} produce an outcome she prefers more; a {\sc scc} is strategy-proof if it is not manipulable. If a {\sc scc} is resolute, which means that it always selects a singleton, there is no ambiguity in understanding whether, for a certain individual, an outcome is better than another. Indeed, that fact can be naturally deduced by her preferences. The well-known Gibbard-Satterthwaite Theorem (Gibbard 1973; Satterthwaite 1975) shows that any resolute {\sc scc} must be manipulable, provided that there are at least three alternatives, each alternative can potentially be an outcome of the {\sc scc}, and the {\sc scc} is not dictatorial.

When a {\sc scc} is not resolute, the definition of strategy-proofness depends on how the outcomes of the {\sc scc}s, which are, in principle, sets of alternatives of any size, are compared by individuals. The way an individual compares sets of alternatives clearly depends on her preferences over alternatives, but there are several reasonable possibilities to specify that dependence. In other words, it is possible to figure out a variety of reasonable extension rules, namely mechanisms that associate with any preference relation on the set of alternatives a preference relation on the set of the nonempty sets of alternatives.\footnote{An analysis of extension rules can be found in Barberà et al.(2004).} The use of different extension rules has led to different definitions of strategy-proofness and several impossibility results have been proved (Pattanaik 1975; G\"ardenfors 1976; Kelly 1977; Barberà 1977a, 1977b; Duggan and Schwartz 2000; Barberà et al. 2001; Taylor 2002; Ching and Zhou 2002; Sato 2008; for a survey, see also Taylor 2005). As observed by Barberà (2011), most contributions establish results for {\sc scc}s that are analogous to the classic impossibility theorem for resolute {\sc scc}s. That suggests that, even when the assumption of resoluteness is removed, there is still no significant room for strategy-proofness.

A notable way to extend individual preferences, which will be considered in this paper, was proposed by Kelly (1977). Specifically, the Kelly extension rule states that an individual prefers one set of alternatives to another if and only if every alternative in the first set is at least as good as every alternative in the second set for that individual.
This requirement is strict and, consequently, makes manipulation difficult for individuals, thereby making impossibility results for strategy-proofness even more significant. At the same time, it also offers a strong and meaningful basis for positive results. In the literature, several possibility and impossibility results involving this extension rule have been established (MacIntyre and Pattanaik 1981; Bandyopadhyay 1982, 1983; Brandt 2015; Brandt et al. 2022a, 2022b).

In this paper, we focus on versions of strategy-proofness for {\sc scc}s that explicitly take into account the informational limitations of individuals. The classic definition of manipulability for resolute {\sc scc}s, as well as its generalization to not necessarily resolute {\sc scc}s, requires that there is an individual who could potentially misreport her preferences based on the knowledge of others' reported preferences.  Thus, the failure of strategy-proofness implies the existence of an individual who has the capability to precisely know others’ preferences. That implicit assumption is definitely unrealistic in many contexts, as it is unlikely that anyone could access such detailed information. As a consequence, violating strategy-proofness may not always be a significant issue.  On the other hand, Nurmi (1987) interestingly observes that an individual might decide to deviate based on a smaller amount of information about others’ preferences. Thus,  manipulability issues become much more significant if the information an individual needs for being profitable to misrepresent her preferences is small enough and easy to obtain. That suggests the possibility to consider notions of strategy-proofness where individual information about the other's preferences is incomplete and only limited to some specific features.

This approach has been developed for resolute {\sc scc}s by Andjiga et al. (2008), Conitzer et al. (2011), Reijngoud and Endriss (2012), and Gori (2021).\footnote{Earlier contributions on preference misrepresentation under restricted beliefs on voters' actions are due to Farquharson (1969), Sengupta (1978, 1980), Bebchuk (1980), Moulin (1981).} Andjiga et al. (2008), given a positive integer $k$, associate with each individual the family of the sets of size $k$ whose elements are lists of the others' preferences, and assume that each individual is able to identify one of these sets containing the true list of others' preferences; Conitzer et al. (2011) generalize the approach by  Andjiga et al. (2008) associating with each individual a general family of sets, called information sets; Reijngoud and Endriss (2012) assume instead that individuals are given only some pieces of information extracted by an opinion poll and described via a so-called poll information function; Gori (2021) generalizes the aforementioned frameworks by letting the information sets depend on individual preferences, as well. The general idea behind the concept of strategy-proofness with limited information is similar for all the described approaches: an individual is not fully aware of the others' preferences but she only knows that the partial preference profile built using the preferences of the others belongs to a specific set; an individual decides to report false preferences if, for every partial preference profile in that set, false preferences cannot make her worse off and, for at least one partial preference profile in that set, they make her better off. Other contributions in the framework of resolute {\sc scc}s are due to Endriss et al. (2016) and Veselova (2020).\footnote{Terzopoulou and Endriss (2019) investigate the problem of manipulation under partial information in the framework of judgment aggregation.}

In this paper, we extend the approach by Gori (2021) to the framework of not necessarily resolute {\sc scc}s. We propose a general definition of manipulability for {\sc scc}s under incomplete information based on two fundamental parameters: an extension rule, which describes how to get a preference over sets of alternatives from any relation over alternatives, and a so-called information function profile, which describes the level of knowledge of each individual (Definition \ref{defpimanip}).\footnote{It is also worth mentioning that the analysis of strategy-proofness under incomplete information for multi-valued voting rule has been recently considered by Tsiaxiras (2021) in a framework different from that of {\sc scc}s.}

After presenting a comparison between the general definition and other standard definitions of strategy-proofness, we focus on the Kelly extension rule and the so-called winner information function profile. The winner information function profile formalizes the assumption that, for a given {\sc scc}, each individual is able to anticipate the set of selected alternatives. Consequently, each individual knows that the partial preference profile formed by the others' preferences has the property that, once combined with her own preferences, it generates a preference profile that leads the {\sc scc} to determine a given outcome. The idea of winner information function profile was introduced and studied by Conitzer et al. (2011) in the framework of resolute {\sc scc}s, and then analyzed in the same framework by Reijngoud and Endriss (2012), Endriss et al. (2016), Veselova and Karabekyan (2023). 

As the main result of the paper, we identify two sensible properties for {\sc scc}s which, if combined, imply manipulability (Theorem \ref{general}). The first property is a monotonicity condition that weakens the set-monotonicity property proposed by Brandt (2015). The second is a very weak sensitivity condition requiring the existence of a preference profile where raising a non-selected alternative in the preferences of a suitable individual affects the outcome. This result implies that, whenever there are at least three alternatives and four individuals, every unanimous and positional {\sc scc} is manipulable (Theorem \ref{main-positional}). 

In addition, we carry out a detailed analysis of manipulability for four well-known {\sc scc}s, namely the Borda {\sc scc}, the plurality {\sc scc}, the negative plurality {\sc scc}, and the Copeland {\sc scc}. We show that these {\sc scc}s exhibit significant differences and that their manipulability properties strongly depend on the number of individuals and the number of alternatives (Theorems \ref{main-borda}, \ref{main-plurality}, \ref{main-neg-plurality}, and \ref{Copeland-main}).

\section{Preliminary results}

Given $k\in\mathbb{N}$, we set $\ldbrack k \rdbrack\coloneq\{x\in \mathbb{N}: x\le k\}$. 
Let $X$ be a nonempty and finite set. We denote by $|X|$ the size of $X$; by $P(X)$ the set of the subsets of $X$; by $P_0(X)$ the set of the nonempty subsets of $X$; by $\mathrm{Sym}(X)$ the set of bijective functions from $X$ to $X$. For $x,y \in X$, the bijection $\psi \in \mathrm{Sym}(X)$ such that $\psi(x) = y$, $\psi(y) = x$ and, for every $z \in X \setminus \{x,y\}$, $\psi(z) = z$, is called the transposition that exchanges $x$ and $y$. A relation on $X$ is a subset of $X^2$, that is, an element of $P(X^2)$. The set of relations on $X$ is denoted by $\mathcal{R}(X)$. 

Let $R\in\mathcal{R}(X)$.  Given $x,y\in X$, we sometimes write $x\succeq_R y$ instead of $(x,y)\in R$; $x\succ_R y$ instead of $(x,y)\in R$ and $(y,x)\not\in R$. Note that $x\succ_R y$ implies $x\succeq_R y$ and $x\neq y$.
We say that $x$ and $y$ are $R$-comparable if at least one between $x\succeq_R y$ and $y\succeq_R x$ holds true.

We say that $R$ is 
\begin{itemize}
\item reflexive if, for every $x\in X$, $x\succeq_R x$;
\item complete if, for every $x,y\in X$, $x\succeq_R y$ or $y\succeq_R x$;
\item transitive if, for every $x,y,z\in X$, $x\succeq_R y$ and  $y\succeq_R z$ imply $x\succeq_R z$;
\item antisymmetric if, for every $x,y\in X$, $x\succeq_R y$ and  $y\succeq_R x$ imply $x=y$;
\item a partial order if $R$ is reflexive, transitive and antisymmetric;
\item a linear order on $X$ if $R$ is complete, transitive and antisymmetric. 
\end{itemize}
Note that, if $R$ is antisymmetric, then, for every $x,y\in X$, $x\succ_R y$  if and only if $x\succeq_R y$ and $x\neq y$;
if $R'\in\mathcal{R}(X)$ is antisymmetric and $R\subseteq R'$, then, for every $x,y\in X$, $x\succ_R y$ implies $x\succ_{R'} y$.

For every $\psi\in\mathrm{Sym}(X)$, we set $\psi R=\{(x,y)\in X^2:(\psi^{-1}(x),\psi^{-1}(y))\in R\}$. Hence, for every $x,y\in X$, $x\succeq_R y$ if and only if $\psi(x)\succeq_{\psi R} \psi(y)$. 
We denote by $\mathcal{L}(X)$ the set of linear orders on $X$. 
Let $R\in \mathcal{L}(X)$ and $|X|=n$ with $n\in\mathbb{N}$. Then $\mathrm{rank}_R:X\to \ldbrack n \rdbrack$ is the bijective function defined, for every $x\in X$, by $\mathrm{rank}_R(x)=|\{y\in X: y\succeq_R x\}|$. 
For every $i\in \ldbrack n \rdbrack$, let $x_i\in X$ be the unique element in $X$ such that $\mathrm{rank}_R(x_i)=i$. Then $R$ is completely determined by the ordered list $(x_i)_{i=1}^n\in X^n$ and thus we represent $R$ by the writing $[x_1,\dots, x_n]$.

\section{Social choice correspondences}

Let us fix two nonempty and finite sets $A$ and $I$ with $|A| \geq 2$ and $|I| \geq 2$. We interpret $A$ as set of alternatives and $I$ as set of individuals.
For every $J\subseteq I$, we denote by $\mathcal{L}(A)^J$ the set of functions from $J$ to $\mathcal{L}(A)$;
the elements of $\mathcal{L}(A)^J$ are called preference profiles of individuals in $J$; any $p\in \mathcal{L}(A)^J$ represents a complete description of the preferences on $A$ of the individuals in $J$ by interpreting, for every $i\in J$, $p(i)\in \mathcal{L}(A)$ as the preferences on $A$ of individual $i$. If, for a given $i\in J$, $p(i)=[x_1,\ldots, x_{|A|}]$,  we refer to $x_1$ as the best alternative for individual $i$ and to $x_{|A|}$ as the worst alternative for individual $i$.
The elements of $\mathcal{L}(A)^I$ are simply called preference profiles.
In order to simplify the reading, in the rest of the paper the elements of $\mathcal{L}(A)^{I}$ will be usually denoted by $p$, possibly with suitable superscripts, and the elements of $\mathcal{L}(A)^{I\setminus\{i\}}$, where $i\in I$, will be usually denoted by $\overline{p}$, possibly with suitable superscripts. If $i\in I$ and $q\in \mathcal{L}(A)$, we denote by $q[i]$ the element of $\mathcal{L}(A)^{\{i\}}$ such that $q[i](i)=q$. Given $i\in I$, $\overline{p}\in \mathcal{L}(A)^{I\setminus\{i\}}$ and $q\in \mathcal{L}(A)$, the writing $(\overline{p},q[i])$ represents the element of $\mathcal{L}(A)^{I}$ such that $(\overline{p},q[i])(i)=q$ and $(\overline{p},q[i])(j)=\overline{p}(j)$ for all $j\in I\setminus \{i\}$.

A social choice correspondence ({\sc scc}) is a function from $\mathcal{L}(A)^I$ to $P_0(A)$. Thus, a social choice correspondence is a procedure that associates with every preference profile a nonempty subset of $A$. A {\sc scc} $F$ is resolute if, for every $p  \in \mathcal{L}(A)^I$, $|F(p)|=1$. For simplicity,  if $F$ is a resolute {\sc scc}, we identify $F(p)$ with the unique element of $F(p)$.

Let us recall now the definition of a very notable family of {\sc scc}s, namely the so-called positional {\sc scc}s. Consider a scoring vector, namely a vector $w=(w_1,\ldots,w_{|A|})\in\mathbb{R}^{|A|}$ such that $w_1\ge w_2\ge \ldots\ge w_{|A|}$ and $w_1>w_{|A|}$. Given  $p \in \mathcal{L}(A)^I$ and  $x \in A$, the $w$-score of $x$ at $p$ is defined by
\[
\mathrm{sc}_w(x,p)\coloneq\sum_{i\in I} w_{\mathrm{rank}_{p(i)}(x)}.
\]
The positional {\sc scc} with scoring vector $w$, or briefly $w$-positional {\sc scc}, is the {\sc scc} that associates, with every $p\in \mathcal{L}(A)^I$, the set
\[
\underset{x \in A}{\mathrm{argmax}}\;\mathrm{sc}_w(x,p).
\]
The Borda {\sc scc}, the plurality {\sc scc}, and the negative plurality {\sc scc}, respectively denoted by $BO$, $PL$ and $NP$, are well-known positional {\sc scc}s respectively defined using the scoring vectors $w_{\mathrm{bo}}=(|A|-1,|A|-2,\ldots,0)$, $w_{\mathrm{pl}}=(1,0,\ldots,0)$ and $w_{\mathrm{np}}=(1,1,\ldots,1,0)$. The $w_{\mathrm{bo}}$-score, the  $w_{\mathrm{pl}}$-score, and the  $w_{\mathrm{np}}$-score are respectively called the Borda score, the plurality score, and the negative plurality score and are simply denoted by $\mathrm{bo}$,  $\mathrm{pl}$ and  $\mathrm{np}$. Observe that, for every  $p\in \mathcal{L}(A)^I$ and $x\in A$, we have
\[
\begin{array}{l}
\mathrm{bo}(x,p)\coloneq\mathrm{sc}_{w_{\mathrm{bo}}}(x,p)=\displaystyle{\sum_{i\in I} \left( |A|-\mathrm{rank}_{p(i)}(x)\right)},\\
\\
\mathrm{pl}(x,p) \coloneq \mathrm{sc}_{w_{\mathrm{pl}}}(x,p)=|\{i \in I : \mathrm{rank}_{p(i)}(x) =1\}|,\\
\vspace{2mm}\\
\mathrm{np}(x,p)\coloneq\mathrm{sc}_{w_{\mathrm{np}}}(x,p)=|\{i \in I : \mathrm{rank}_{p(i)}(x) \neq |A|\}|.\\
\end{array}
\]
The Copeland {\sc scc} is another  well studied {\sc scc} (Copeland 1951; Fishburn 1977). Let $p \in \mathcal{L}(A)^I$. For every $x,y \in A$, we set $c_p(x,y) \coloneq |\{i \in I : x \succ_{p(i)} y\}|$.
For every $x \in A$, we also set 
\begin{align*}
    w_p(x) \coloneq |\{y \in A \setminus \{x\} : c_p(x,y) > c_p(y,x)\}|;\\
    l_p(x) \coloneq |\{y \in A \setminus \{x\} : c_p(x,y) < c_p(y,x)\}|,
\end{align*}
and we define the Copeland score of $x$ at $p$ as   
\[
\mathrm{co}(x,p) \coloneq w_p(x) - l_p(x).
\]
The Copeland {\sc scc}, denoted by $CO$, is the the {\sc scc} defined, for every $p \in \mathcal{L}(A)^I$, by
\[
CO(p) \coloneq \underset{x \in A}{\mathrm{argmax}}\;\mathrm{co}(x,p).
\]
It is easily checked that $BO$, $PL$, $NP$ and $CO$ are not resolute unless $|A|=2$ and $|I|$ is odd.

A {\sc scc} $F$ is called unanimous if, for every $p\in \mathcal{L}(A)^I$ and $x\in A$, the fact that $\mathrm{rank}_{p(i)}(x)=1$ for all $i\in I$  implies $F(p)=\{x\}$. It is a simple exercise to prove that a positional {\sc scc} with a scoring vector $w$ is unanimous if and only if $w_1>w_2$ and the Copeland {\sc scc} is unanimous.

\section{Extension rules}

Following Barberà et al. (2004), an extension rule is a function $\mathbf{E}$  from $\mathcal{L}(A)$ to $\mathcal{R}(P_0(A))$ such that
\begin{equation}\label{def-ext}\hbox{for every} \ q\in \mathcal{L}(A)\  \hbox{and}\  x,y \in A, \{x\} \succeq_{\mathbf{E}(q)} \{y\} \ \hbox{if and only if}\  x \succeq_q y.
\end{equation}
If $\mathbf{E}$ is an extension rule and $q\in\mathcal{L}(A)$ represents the preferences on $A$ of an individual, we interpret the relation $\mathbf{E}(q)$ as a description of the preferences of that individual on the set $P_0(A)$. The problem of reasonably extending the preferences of an individual from a set of alternatives to the set of subsets of alternatives is crucial and largely investigated, and there are a variety of extension rules considered in the literature, each of them based on a specific rationale. 

In this paper, we focus on the well-known Kelly extension rule (Kelly 1977), defined, for every $q \in \mathcal{L}(A)$, as 
\[
\mathbf{K}(q) \coloneq \{(B,C)\in P_0(A)^2:B=C\}\cup \left\{(B,C)\in P_0(A)^2 : x \succeq_q y \mbox{ for all } x \in B \mbox{ and } y \in C \right\}.
\]
It is simple to prove that $\mathbf{K}$ is actually an extension rule. Moreover, for every $q\in\mathcal{L}(A)$, $\mathbf{K}(q)$ is a partial order on $P_0(A)$ that, in general, is far from being complete. Proposition \ref{properties} in Appendix \ref{appendix-A} collects some basic properties of $\mathbf{K}$ that will be used throughout the paper without reference.

\section{Manipulation of {\sc scc}s}

Let us recall the well-known definitions of manipulability and strategy-proofness for resolute {\sc scc}.

\begin{definition} \label{defmanip-scf}
Let $F$ be a resolute {\sc scc}. $F$ is manipulable if there exist  $i \in I$, $q,q' \in \mathcal{L}(A)$ and $\overline{p} \in \mathcal{L}(A)^{I\setminus\{i\}}$  such that
$F(\overline{p},q'\left[i\right]) \succ_{q} F(\overline{p},q\left[i\right])$.
$F$ is strategy-proof if it is not manipulable.
\end{definition}

The following definition, which corresponds to Definition 2 in Brandt and Brill (2011), extends the standard concept of strategy-proofness, originally introduced for resolute {\sc scc}s, to the case of {\sc scc}s that are not necessarily resolute. 

\begin{definition} \label{defmanip-scc}
Let $F$ be a {\sc scc} and $\mathbf{E}$ be an extension rule. $F$ is $\mathbf{E}$-manipulable if there exist $i \in I$, $q,q' \in \mathcal{L}(A)$ and $\overline{p} \in \mathcal{L}(A)^{I\setminus\{i\}}$  such that
$F(\overline{p},q'\left[i\right]) \succ_{\mathbf{E}(q)} F(\overline{p},q\left[i\right])$.
$F$ is $\mathbf{E}$-strategy-proof if it is not $\mathbf{E}$-manipulable.
\end{definition}

Note that if $F$ is a resolute {\sc scc} and $\mathbf{E}$ an extension rule, then, by  \eqref{def-ext}, $F$ is $\mathbf{E}$-manipulable if and only if $F$ is manipulable.

In recent years, some authors have started considering the problem of manipulations of resolute social choice correspondences under the assumptions that individuals are not able to exactly know the preferences of the others, as assumed in Definitions \ref{defmanip-scf} and \ref{defmanip-scc}, but only have limited information about them (Andjiga et al. 2008, Conitzer et al. 2011; Reijngoud and Endriss 2012; Endriss et al. 2016; Veselova 2020; Gori 2021). An effective way to model the different levels of information of individuals is the use of the so-called information function profiles, which are introduced in Gori (2021), and generalize both the concept of information sets by Conitzer et al. (2011) and the one of poll information function by  Reijngoud and Endriss (2012).

Let $i\in I$. An information function for individual $i$ is a function $\Omega_i: \mathcal{L}(A) \to P_0(P_0(\mathcal{L}(A)^{I\setminus \{i\}}))$. Thus, $\Omega_i$ associates  a nonempty set of nonempty subsets of $\mathcal{L}(A)^{I\setminus \{i\}}$ with each $q\in \mathcal{L}(A)$. The idea is that if individual $i$ has preference $q$, then the type of information that she has about the preferences of the others realizes in her capability to identify a set $\omega \in \Omega_i(q)$ which surely contains among its elements the preference profile of the other individuals.
An information function profile is a list $\Omega=(\Omega_i)_{i\in I}$ that collects the information functions of all the individuals. Let us now present three basic examples of information function profiles, which have been first introduced and investigated by Conitzer et al. (2011). 

The {\it complete information function profile}, denoted by $\Omega^{c}$, is defined, for every $i\in I$ and $q\in \mathcal{L}(A)$, by
\[
\Omega^{c}_i(q)\coloneq\Big\{\{\overline{p}\}: \overline{p}\in \mathcal{L}(A)^{I\setminus \{i\}}\Big\}.
\]
When this information function profile is considered, every $\omega \in \Omega^{c}_i(q)$ consists of exactly one preference profile of individuals in $I\setminus \{i\}$. Thus, $\Omega^{c}$ describes the situation where each individual exactly knows the preferences of the other individuals.

The {\it zero information function profile}, denoted by $\Omega^{0}$, is defined, for every $i\in I$ and $q\in \mathcal{L}(A)$, by
\[
\Omega^0_i(q)\coloneq\Big\{\mathcal{L}(A)^{I\setminus \{i\}}\Big\}.
\]
In this case,  every $\omega \in \Omega^{0}_i(q)$ consists of the whole set $\mathcal{L}(A)^{I\setminus \{i\}}$. Thus, $\Omega^{0}$ describes the situation where each individual only knows the obvious fact that the preference profile of individuals in $I\setminus\{i\}$ belongs to $\mathcal{L}(A)^{I\setminus \{i\}}$.

Note that, for every $q,q'\in \mathcal{L}(A)$, $\Omega^{c}(q)=\Omega^{c}(q')$ and  $\Omega^{0}(q)=\Omega^{0}(q')$. Thus, for those two information function profiles the dependence on $q$ is fictitious. We now introduce a significant example where the dependence of $q$ does play a role.

Given a {\sc  scc} $F$, the $F$-{\it winner information function profile}, denoted by $\Omega^{F}$, is defined, for every $i\in I$ and $q\in \mathcal{L}(A)$, by
\[
\Omega^F_i(q)\coloneq\Big\{\{\overline{p}\in \mathcal{L}(A)^{I\setminus \{i\}}: F(\overline{p},q[i])=X\}: X\in \mathrm{Im}(F)\Big\}\setminus\{\varnothing\}.
\]
In this case, each $\omega\in \Omega^F_i(q)$ is formed by all the preference profiles of individuals in $I\setminus\{i\}$ that, completed with $q$ as the preference of individual $i$, give the same outcome. Thus, $\Omega^{F}$ describes the situation where the information that each individual has is the knowledge of the final outcome obtained applying $F$ to the preferences of the individuals in the society. Note that $\Omega^F$ is actually an information function profile, that is, for every $i\in I$ and  $q\in \mathcal{L}(A)$, $\Omega^F_i(q)\neq\varnothing$. Indeed, pick $\overline{p}^*\in \mathcal{L}(A)^{I\setminus \{i\}}$  and let $X^*=F(\overline{p}^*,q[i])$. Then $X^*\in \mathrm{Im}(F)$ and  $\overline{p}^*\in\{\overline{p}\in \mathcal{L}(A)^{I\setminus \{i\}}: F(\overline{p}^*,q[i])=X^*\}\neq\varnothing$. Thus, $\{\overline{p}\in \mathcal{L}(A)^{I\setminus \{i\}}: F(\overline{p}^*,q[i])=X^*\}\in \Omega^F_i(q)\neq\varnothing$.

The following definition, which corresponds to Definition 3 in Gori (2021), describes the meaning of manipulability and strategy-proofness for resolute {\sc scc}s when individuals' information about the preferences of the others is described by a suitable information function profile. 

\begin{definition}\label{manip-scf-inf}
Let $F$ be a resolute {\sc scc} and $\Omega$ be an information function profile. $F$ is $\Omega$-manipulable if there exist $i \in I$, $q,q'\in \mathcal{L}(A)$ and $\omega\in \Omega_i(q)$ such that
\begin{itemize}
\item there exists $\overline{p}'\in \omega$ with $F(\overline{p}',q'\left[i\right]) \succ_{q} F(\overline{p}',q\left[i\right])$;
\item for every $\overline{p} \in \omega$, we have $F(\overline{p},q'\left[i\right]) \succeq_{q} F(\overline{p},q\left[i\right])$.
\end{itemize} 
$F$ is $\Omega$-strategy-proof if it is not $\Omega$-manipulable.
\end{definition}

By combining Definitions \ref{defmanip-scc} and \ref{manip-scf-inf} it is possible to introduce the concepts of manipulability and strategy-proofness for {\sc scc}s that are not necessarily resolute, under the assumption of incomplete information. 

\begin{definition} \label{defpimanip}
Let $F$ be a {\sc scc}, $\mathbf{E}$ be an extension rule and $\Omega$ be an information function profile. $F$ is $\Omega$-$\mathbf{E}$-manipulable if there exist $i \in I$, $q,q'\in \mathcal{L}(A)$ and $\omega\in \Omega_i(q)$ such that
\begin{itemize}
\item there exists $\overline{p}'\in \omega$ with $F(\overline{p}',q'\left[i\right]) \succ_{\mathbf{E}(q)} F(\overline{p}',q\left[i\right])$; 
\item for every $\overline{p} \in \omega$, we have $F(\overline{p},q'\left[i\right]) \succeq_{\mathbf{E}(q)} F(\overline{p},q\left[i\right])$.
\end{itemize} 
$F$ is $\Omega$-$\mathbf{E}$-strategy-proof if it is not $\Omega$-$\mathbf{E}$-manipulable.
\end{definition}

Thus, $F$ is $\Omega$-$\mathbf{E}$-strategy proof if, every time an individual $i$, whose preferences are described by $q\in \mathcal{L}(A)$, knows that the preferences of the others are surely described by some element of $\omega\in \Omega_i(q)$ and observes that for an element $\overline{p}'\in \omega$ it is convenient for her to report the false preferences described by $q'$, then there exists another element $\overline{p}\in \omega$ for which 
$F(\overline{p},q'\left[i\right]) \not\succeq_{\mathbf{E}(q)} F(\overline{p},q\left[i\right])$, that is, one of the two following situation holds true:
\begin{itemize}
    \item $F(\overline{p},q\left[i\right]) \succ_{\mathbf{E}(q)} F(\overline{p},q'\left[i\right])$, meaning that, for the element $\overline{p}\in \omega$, she would be better off if she told the truth;
    \item $F(\overline{p},q\left[i\right])$ and $ F(\overline{p},q'\left[i\right])$ are not $\mathbf{E}(q)$-comparable, meaning that, for the element $\overline{p}\in \omega$, she is not able to compare the two sets $F(\overline{p},q\left[i\right])$ and $ F(\overline{p},q'\left[i\right])$.
\end{itemize}
Note that
\begin{equation}\label{1}
    \mbox{$F$ is $\Omega^{c}$-$\mathbf{E}$-manipulable if and only if $F$ is $\mathbf{E}$-manipulable}.
\end{equation} 
Moreover, if $F$ is resolute, $F$ is $\Omega$-$\mathbf{E}$-manipulable if and only if $F$ is $\Omega$-manipulable, and $F$ is $\Omega^{c}$-$\mathbf{E}$-manipulable if and only if $F$ is manipulable.

We end this section with some propositions that allow to deduce the $\Omega'$-$\mathbf{E}'$-manipulability of a social choice correspondence from its 
$\Omega$-$\mathbf{E}$-manipulability, provided that suitable properties of the information function profiles $\Omega$ and $\Omega'$ and the extension rules $\mathbf{E}$ and $\mathbf{E}'$ are satisfied.

Given two extension rules $\mathbf{E}$ and $\mathbf{E}'$, we say that $\mathbf{E}$ is a refinement of $\mathbf{E}'$ if, for every $q \in \mathcal{L}(A)$, $\mathbf{E}(q) \subseteq \mathbf{E}'(q)$; if $\mathbf{E}$ is a refinement of $\mathbf{E}'$, we write $\mathbf{E} \subseteq \mathbf{E}'$.

\begin{proposition}\label{lemimpepiman}
Let $F$ be a {\sc scc}, $\mathbf{E}$ and $\mathbf{E}'$ be extension rules with  $\mathbf{E}\subseteq \mathbf{E}'$, and $\Omega$ be an information function profile. Assume that $F$ is  $\Omega$-$\mathbf{E}$-manipulable and that, for every $q\in\mathcal{L}(A)$, $\mathbf{E}'(q)$ is antisymmetric. Then $F$ is $\Omega$-$\mathbf{E}'$-manipulable.
\end{proposition}

\begin{proof}
Since $F$ is $\Omega$-$\mathbf{E}$-manipulable, we have that there exist $i \in I$, $q,q'\in \mathcal{L}(A)$, $\omega\in \Omega_i(q)$ and $\overline{p}'\in \omega$ such that $F(\overline{p}',q'\left[i\right]) \succ_{\mathbf{E}(q)} F(\overline{p}',q\left[i\right])$;
for every $\overline{p} \in \omega$, we have $F(\overline{p},q'\left[i\right]) \succeq_{\mathbf{E}(q)} F(\overline{p},q\left[i\right])$.
Since $\mathbf{E} \subseteq \mathbf{E}'$ we immediately have that, for every $\overline{p} \in \omega$, $F(\overline{p},q'\left[i\right]) \succeq_{\mathbf{E}'(q)} F(\overline{p},q\left[i\right])$. Since $F(\overline{p}',q'\left[i\right]) \succ_{\mathbf{E}(q)} F(\overline{p}',q\left[i\right])$ and  $\mathbf{E}'(q)$ is antisymmetric, we have $F(\overline{p}',q'\left[i\right]) \succ_{\mathbf{E}'(q)} F(\overline{p}',q\left[i\right])$. Then, $F$ is $\Omega$-$\mathbf{E}'$-manipulable.
\end{proof}

Given $\Omega$ and $\Omega'$ information function profiles, we say that $\Omega$ is at least as informative as $\Omega'$, and we write $\Omega \trianglerighteq \Omega'$, if, for every $i \in I$, $q \in \mathcal{L}(A)$ and $\omega' \in \Omega'_i(q)$, there exists $\mathcal{A} \subseteq \Omega_i(q)$ with $\mathcal{A} \neq \varnothing$ such that $\omega' = \bigcup_{\omega \in \mathcal{A}}\omega$ (Gori 2021, Definition 4). 
Note that, for every information function profile $\Omega$, we have $\Omega^c\trianglerighteq \Omega$.

\begin{proposition}
\label{lemimppisp}
	Let $F$ be a {\sc scc}, $\mathbf{E}$ be an extension rule, and $\Omega$ and $\Omega'$ be information function profiles such that $\Omega \trianglerighteq \Omega'$. Assume that $F$ is $\Omega'$-$\mathbf{E}$-manipulable. Then $F$ is $\Omega$-$\mathbf{E}$-manipulable.
\end{proposition}

\begin{proof}
Since $F$ is $\Omega'$-$\mathbf{E}$-manipulable, then there exist $i \in I$, $q,q'\in \mathcal{L}(A)$ and $\omega' \in \Omega'_i(q)$ such that, for every $\overline{p} \in \omega'$, we have $F(\overline{p},q'\left[i\right]) \succeq_{\mathbf{E}(q)} F(\overline{p},q\left[i\right])$, and there exists $\overline{p}'\in \omega'$ such that $F(\overline{p}',q'\left[i\right]) \succ_{\mathbf{E}(q)} F(\overline{p}',q\left[i\right])$. Since $\Omega \trianglerighteq \Omega'$, we have $\omega' = \bigcup_{\omega \in \mathcal{A}} \omega$, for some $\mathcal{A}\subseteq \Omega_i(q)$ with  $ \mathcal{A}  \neq\varnothing$. Then there exists $\omega^* \in \mathcal{A}$ such that $\overline{p}' \in \omega^* $. Since $\omega^*  \subseteq \omega'$ we have that, for every $\overline{p} \in \omega^* $, $F(\overline{p},q'\left[i\right]) \succeq_{\mathbf{E}(q)} F(\overline{p},q\left[i\right])$. Since we know that $F(\overline{p}',q'\left[i\right]) \succ_{\mathbf{E}(q)} F(\overline{p}',q\left[i\right])$ and $\overline{p}' \in \omega^* $, we conclude that $F$ is $\Omega$-$\mathbf{E}$-manipulable.
\end{proof} 

It is also worth mentioning the two following corollaries that are immediate consequences of \eqref{1}, Propositions \ref{lemimpepiman} and \ref{lemimppisp}, and the properties of $\Omega^c$.

\begin{corollary} 
Let $F$ be a {\sc scc}, and $\mathbf{E}$ and $\mathbf{E}'$ be extension rules with $\mathbf{E} \subseteq \mathbf{E}'$. Assume that $F$ is $\mathbf{E}$-manipulable and that, for every $q \in \mathcal{L}(A)$, $\mathbf{E}'(q)$ is antisymmetric. Then $F$ is $\mathbf{E}'$-manipulable. 
\end{corollary}

\begin{corollary}
    Let $F$ be a {\sc scc}, $\mathbf{E}$ an extension rule and $\Omega$ an information function profile. Assume that $F$ is $\Omega$-$\mathbf{E}$-manipulable. Then $F$ is $\mathbf{E}$-manipulable.
\end{corollary}

\section{Main results}

Suppose that one is interested in using a certain {\sc scc} $F$, and suppose that the extension rule $\mathbf{E}$ carefully describes the way individuals extend their preferences from $A$ to $P_0(A)$. Given an information function profile $\Omega$, one and only one of the following three situations can occur:
\begin{itemize}
\item $F$ is $\mathbf{E}$-strategy-proof (and then $\Omega$-$\mathbf{E}$-strategy-proof), 
\item $F$ is $\mathbf{E}$-manipulable and $\Omega$-$\mathbf{E}$-strategy-proof, 
\item $F$ is $\Omega$-$\mathbf{E}$-manipulable (and then $\mathbf{E}$-manipulable).
\end{itemize}
If $F$ is $\mathbf{E}$-strategy-proof, then we know that no possible strategic behavior may be implemented; if $F$ is $\mathbf{E}$-manipulable and $\Omega$-$\mathbf{E}$-strategy-proof, then we know that full information may cause a possible deviation but if the level of information of individuals corresponds at most to the one described by $\Omega$ nobody has an incentive to report false preferences; if $F$ is $\Omega$-$\mathbf{E}$-manipulable, then individuals have an incentive to report false preferences even if their level of information corresponds at least to the one described by $\Omega$.

In the rest of the paper, we focus on the Kelly extension rule and the winner information function profile. We first identify two properties for {\sc scc}s that combined guarantee sufficient conditions for manipulability. We then apply this result to get information about positional and unanimous {\sc scc}s. Finally we develop a full-fledged analysis for the Borda {\sc scc}, the plurality {\sc scc}, and the negative plurality {\sc scc}. Those results are described in Sections \ref{sec-teo-gen} and \ref{positional}.

\subsection{Sufficient conditions for \texorpdfstring{$\Omega^F$}{ΩF}-\texorpdfstring{$\mathbf{K}$}{K}-manipulability}\label{sec-teo-gen}

In this section, we propose a result that gives conditions on a {\sc scc} $F$ that are sufficient for its $\Omega^F$-$\mathbf{K}$-manipulability. We start by defining a property of monotonicity that, to the best of our knowledge, is new.

\begin{definition}\label{wsm}
Let $F$ be a {\sc scc}. $F$ satisfies weak set-monotonicity ({\sc wsm}) if, for every $i\in I$, $\overline{p}\in \mathcal{L}(A)^{I\setminus\{i\}}$, $q\in\mathcal{L}(A)$ and $x,y,z\in A$ such that 
\begin{itemize}
\item $F(\overline{p},q[i])=\{z\}$,  
\item $\mathrm{rank}_{q}(x)+1=\mathrm{rank}_{q}(y)<\mathrm{rank}_{q}(z)$,
\end{itemize}
we have that $F(\overline{p},\psi q\left[i\right])\subseteq \{z,y\}$, where $\psi\in \mathrm{Sym}(A)$ is the transposition that exchanges $x$ and $y$.
\end{definition}

Thus, a {\sc scc} satisfies {\sc wsm} if, for every preference profile for which exactly one winner is selected, if an individual exchanges in her preferences the position of two consecutive alternatives that are ranked above the winner, then one of the following facts occurs: the set of winners does not change; the raised alternative becomes the unique winner; the raised alternative becomes a winner together with the former winner.  Note that if $|A|=2$, then every {\sc scc} satisfies  {\sc wsm}.
Moreover,  as proved in Proposition \ref{pos-wum}, every positional {\sc scc} satisfies  {\sc wsm} (see Appendix \ref{appendix-B}).

Several notions of monotonicity are present in the literature. The reader may refer to Sanver and Zwicker (2012) for a detailed discussion of these properties for {\sc scc}s. Our definition can be seen in the spirit of the stability concept introduced by Campbell et al. (2018) for resolute {\sc scc}s, in which certain improvements of a non-selected alternative can cause this alternative to be selected.

Brandt (2015) introduces instead a more static property, called set-monotonicity, and proves that it implies $\mathbf{K}$-strategy-proofness (Brant 2015, Theorem 1).

\begin{definition}\label{set-mon}
Let $F$ be a {\sc scc}. $F$ satisfies  set-monotonicity ({\sc sm}) if, for every $i\in I$, $\overline{p}\in \mathcal{L}(A)^{I\setminus\{i\}}$, $q\in\mathcal{L}(A)$ and $x,y\in A$ such that 
\begin{itemize}
\item $x\not\in F(\overline{p},q[i])$,  
\item $\mathrm{rank}_{q}(x)+1=\mathrm{rank}_{q}(y)$,
\end{itemize}
we have that $F(\overline{p},\psi q\left[i\right])=F(\overline{p},q[i])$, where $\psi\in \mathrm{Sym}(A)$ is the transposition that exchanges $x$ and $y$.
\end{definition}

Thus, a {\sc scc} satisfies {\sc sm} if, for every preference profile, whenever an individual exchanges two consecutive alternatives in her preferences, where the higher-ranked alternative is not part of the choice set, this change does not affect the choice set. Set monotonicity is a strong requirement, and well-known {\sc scc}s fail to fulfill it. Assume, for example, that $A = \{x,y,z\}$ and $I = \ldbrack 4 \rdbrack$, and take $p,p' \in \mathcal{L}(A)^I$ such that
\[
p(1) = p'(1) =[z,x,y],\;  p(2) =p'(2) = [y,z,x],
\]
\[
p(3) =p'(3) = [z,y,x],\; p(4) = [x,y,z],\; p'(4) = [y,x,z].
\]
A computation shows that we have that $CO(p) = \{z\}$ and $CO(p')=\{y,z\}$, which shows that  the Copeland {\sc scc} can fail {\sc sm}.

The following proposition motivates the use of the term weak set-monotonicity for the property introduced in Definition \ref{wsm}.

\begin{proposition}\label{implies}
Let $F$ be a {\sc scc}. If $F$ satisfies {\sc sm}, then it satisfies {\sc wsm}.
\end{proposition}

\begin{proof}
Assume that $F$ satisfies {\sc sm}. Consider $i\in I$, $\overline{p}\in \mathcal{L}(A)^{I\setminus\{i\}}$, $q\in\mathcal{L}(A)$ and $x,y,z\in A$ such that 
$F(\overline{p},q[i])=\{z\}$ and $\mathrm{rank}_{q}(x)+1=\mathrm{rank}_{q}(y)<\mathrm{rank}_{q}(z)$.
Then $x\not\in F(\overline{p},q[i])$ and thus, using {\sc sm}, we have $F(\overline{p},\psi q\left[i\right])=F(\overline{p},q[i])=\{z\}\subseteq \{z,y\}$. That proves that $F$ satisfies {\sc wsm}, as well. 
\end{proof}

We emphasize that, from a conceptual point of view, moving from {\sc sm} to {\sc wsm} represents a significant step. Indeed, while {\sc sm} requires, among other things, the output to remain fixed under improvements of non-winning alternatives, {\sc wsm} allows such improvements, in certain circumstances, to influence the output.

Let us now introduce a further new property, which basically states that individuals should have a minimal degree of decision power. More exactly, there exists at least an individual and  a special preference configuration such that a change in her preferences actually influences the outcome. This seems to be a quite natural and weak requirement to which it is hard to renounce.

\begin{definition}
Let $F$ be a {\sc scc}. We say that $F$ satisfies upward sensitivity ({\sc us}) if
there exist $i\in I$, $\overline{p}\in \mathcal{L}(A)^{I\setminus\{i\}}$, $q\in\mathcal{L}(A)$
and $x,y,z \in A$ such that
\begin{itemize}
\item $F(\overline{p},q[i])=\{z\}$,
\item $\mathrm{rank}_{q}(x)+1=\mathrm{rank}_{q}(y)<\mathrm{rank}_{q}(z)$,
\item $y \in F(\overline{p},\psi q\left[i\right])$, where $\psi\in \mathrm{Sym}(A)$ is the transposition that exchanges $x$ and $y$.
\end{itemize}
\end{definition}

Thus, a {\sc scc} satisfies {\sc us} if there exists a preference profile for which exactly one winner is selected, and an individual who, by exchanging the position of two consecutive alternatives that are not winners and are ranked above the winner in her preferences, makes the set of winners include the raised alternative.
The property of being {\sc us} naturally complements {\sc wsm}, as it assures that, in at least one scenario, the raised alternative is indeed included in the output.

It is worth noting that {\sc wsm} and {\sc us} are independent conditions. Indeed, if $|A|=2$, any {\sc scc} satisfies {\sc wsm} but fails to satisfy {\sc us}. On the other hand, there exists a {\sc scc} that satisfies {\sc us} but not {\sc wsm}. Consider $A=\{x,y,z\}$ and $I=\ldbrack 7 \rdbrack$.
Let $F$ be the {\sc scc} defined, for every $p\in \mathcal{L}(A)^I$, as $F(p)= \{x^*\}$, where $x^*$ is the alternative ranked first more times than all others, if such an alternative exists, and otherwise is the alternative first ranked by individual $1$. Clearly, $F$ is resolute. Let $p\in \mathcal{L}(A)^{I}$ be such that
\[
p(1)=p(7)=[x,y,z],\;p(2)=p(3)=p(4)=[z,x,y],\;p(5)=p(6)=[y,z,x].
\]
Then $F(p)=\{z\}$. If $p'$ is the preference profile obtained by $p$ by exchanging $x$ and $y$ in the preferences of individual $7$, we have that  
$F(p')=\{x\}\not\subseteq \{z,y\}$ and this shows that $F$ does not satisfy {\sc wsm}. Let now $p''$ be the preference profile obtained by 
$p'$ by exchanging $y$ and $z$ in the preferences of individual $6$. Thus, passing from $p'(6)$ to $p''(6)$, $z$ improves its ranking. Since $z\in F(p'')=\{z\}$, we conclude that $F$ satisfies {\sc us}.

In the next result, we show that {\sc us} and {\sc sm} are mutually exclusive properties, whereas {\sc us} and {\sc wsm} are not.
\begin{proposition}\label{application} 
There exists no {\sc scc} satisfying both {\sc sm} and {\sc us}. There are {\sc scc}s satisfying both {\sc wsm} and {\sc us}.
\end{proposition}
\begin{proof} 
We show that if $F$ is a {\sc scc} satisfying {\sc us}, then $F$ does not satisfy {\sc sm}. Let $F$ be a {\sc scc} satisfying {\sc us}. 
Then there exist $i\in I$, $\overline{p}\in \mathcal{L}(A)^{I\setminus\{i\}}$, $q\in\mathcal{L}(A)$,
and $x,y,z \in A$ such that
\begin{itemize}
\item $F(\overline{p},q[i])=\{z\}$,
\item $\mathrm{rank}_{q}(x)+1=\mathrm{rank}_{q}(y)<\mathrm{rank}_{q}(z)$,
\item $y \in F(\overline{p},\psi q\left[i\right])$, where $\psi\in \mathrm{Sym}(A)$ is the transposition that exchanges $x$ and $y$.
\end{itemize}
Then, $x,y, z$ are distinct. In particular, $x\notin F(\overline{p},q[i]).$ Now 
$F(\overline{p},\psi q\left[i\right])\neq F(\overline{p},q[i])$ because $y\in F(\overline{p},\psi q\left[i\right])$ while $y\notin F(\overline{p},q[i])=\{z\}.$ Thus, $F$ fail to satisfy {\sc sm}.

The fact that there are {\sc scc}s satisfying both {\sc wsm} and {\sc us} is a consequence of Propositions \ref{pos-wum} and \ref{us-unanimous} (see Appendix \ref{appendix-B}). 
\end{proof}

While set-monotonic {\sc scc}s can be appreciated since they are $\mathbf{K}$-strategy-proof, the fact that they never satisfy a natural property such as upward sensitivity can be regarded, in some situations,  as a serious drawback.

We now state and prove the main result of this section. It shows that weakening {\sc sm} by replacing it with {\sc wsm}, and further adding  {\sc us}, entails not only $\mathbf{K}$-manipulability but even $\Omega^F$-$\mathbf{K}$-manipulability.

\begin{theorem} \label{general}
Let $F$ be a {\sc scc}. Assume that $F$ satisfies {\sc wsm} and {\sc us}. Then $F$ is $\Omega^F$-$\mathbf{K}$-manipulable.
\end{theorem}

\begin{proof}
In order to prove that $F$ is $\Omega^F$-$\mathbf{K}$-manipulable, we need to show that there exist  $i \in I$, $q,q' \in \mathcal{L}(A)$ and $\omega \in \Omega^F_i(q)$ such that, for every $\overline{p} \in \omega$, we have $F(\overline{p},q'[i]) \succeq_{\mathbf{K}(q)} F(\overline{p},q[i])$ and there exists $\overline{p}' \in \omega$ such that $F(\overline{p}',q'[i]) \succ_{\mathbf{K}(q)} F(\overline{p}',q\left[i\right])$.

Since $F$ satisfies {\sc us} there exist
$i\in I$, $\overline{p}'\in \mathcal{L}(A)^{I\setminus\{i\}}$, $q\in\mathcal{L}(A)$
and $x,y,z \in A$ such that
\begin{itemize}
\item $F(\overline{p}',q[i])=\{z\}$,
\item $\mathrm{rank}_{q}(x)+1=\mathrm{rank}_{q}(y)<\mathrm{rank}_{q}(z)$,
\item $y \in F(\overline{p}',\psi q[i])$, where $\psi$ is the transposition that exchanges $x$ and $y$.
\end{itemize}
In particular, since $F$ satisfies {\sc wsm}, we have that $F(\overline{p}',\psi q[i])=\{y,z\}$ or $F(\overline{p}',\psi q\left[i\right])=\{y\}$.

Let now $i\in I$ and $q\in\mathcal{L}(A)$ be the ones previously considered, and set $q'= \psi q$ and $\omega = \{\overline{p} \in \mathcal{L}(A)^{I\setminus\{i\}} : F(\overline{p},q[i]) = \{z\}\}$. 
Since $F(\overline{p}',q[i])=\{z\}$, we have that $\omega\neq \varnothing$ and so $\omega\in \Omega^F_i(q)$.
Of course, we have that $F(\overline{p}',q'[i]) \succ_{\mathbf{K}(q)} \{z\}= F(\overline{p}',q[i])$, because $y\succ _q z$.
We are then left with proving that, for every $\overline{p} \in \omega$, we have $F(\overline{p}, q'[i]) \succeq_{\mathbf{K}(q)} F(\overline{p},q[i])$. Let $\overline{p} \in \omega$. Thus, we have $F(\overline{p}, q[i]) =\{z\}$. Moreover, by {\sc wsm}, we have $F(\overline{p}, \psi  q[i]) = F(\overline{p}, q'[i])\in \{\{z\},\{y\}, \{y,z\}\}$. Since, for every $B\in \{\{z\},\{y\}, \{y,z\}\}$ we have $ B \succeq_{\mathbf{K}(q)} \{z\}$, we conclude that $F(\overline{p}, q'[i]) \succeq_{\mathbf{K}(q)} F(\overline{p}, q[i])$.
\end{proof}

When $|A|=2$ no {\sc scc} can satisfy both {\sc wsm} and {\sc us}. Thus, Theorem \ref{general} is not informative when  $|A|=2$. As we will see, it is instead a useful tool to analyze unanimous positional {\sc scc}s when $|A|\ge 3$.

\subsection{Positional social choice correspondences}\label{positional}

A complete analysis of the property of $\Omega^F$-$\mathbf{K}$-manipulability for a general positional {\sc scc} seems prohibitive due to insidious arithmetical technicalities. However, many things can be said, especially concerning the most commonly used {\sc scc}s within such a class. First, let us observe that the case with two alternatives is trivial. Indeed, as a consequence of Proposition \ref{qualified}, stated and proved in Appendix \ref{sec-A2}, we have that each positional {\sc scc}, which coincides with the simple majority rule, is $\mathbf{K}$-strategy-proof and then $\Omega^F$-$\mathbf{K}$-strategy-proof. Restricting our attention to the case with at least three alternatives, as a consequence of Theorem \ref{general}, we obtain the following result whose proof is given in Appendix \ref{appendix-B}.

\begin{theorem}\label{main-positional}
Assume that $|A|\ge 3$ and $|I|\ge 4$. If $F$ is an unanimous positional {\sc scc}, then $F$ is $\Omega^{F}$-$\mathbf{K}$-manipulable.  
\end{theorem}

We stress that Theorem \ref{main-positional} is analogous to Theorem 3 in Reijngoud and Endriss (2012), where the authors focus on unanimous and positional {\sc scc}s made resolute by an agenda for breaking ties. Also, the proofs of those theorems share some similarities. However, neither of the two is a corollary of the other. In fact, it is well-established in the literature that breaking ties can be performed in several ways, respecting or not respecting specific properties (see, for instance, Bubboloni and Gori 2016, 2021), and that there is then no obvious way to transfer a result for resolute {\sc scc}s to  general {\sc scc}s and conversely. Note also that, as an immediate consequence of Theorem 1 in Brandt (2015) and of Theorem \ref{main-positional}, if $|A|\geq 3$ and $|I|\geq 4$, every unanimous positional {\sc scc} fail to satisfy set-monotonicity. 
The analysis of positional rules when $|I|\in\{2,3\}$ or when they are not unanimous is more difficult and seems to lead to a variety of different situations. That clearly emerges from the detailed analysis of the Borda, the plurality and the negative plurality {\sc scc}, as summarized in the following three theorems.

\begin{theorem}\label{main-borda}
If $|A|\ge 3$, then $BO$ is $\Omega^{BO}$-$\mathbf{K}$-manipulable.
\end{theorem}

\begin{theorem}\label{main-plurality}
If $|A|\ge 3$ and $|I|\in\{2,3\}$, then $PL$ is $\mathbf{K}$-strategy-proof. If $|A|\ge 3$ and $|I|\ge 4$, then $PL$ is $\Omega^{PL}$-$\mathbf{K}$-manipulable.
\end{theorem}

\begin{theorem}\label{main-neg-plurality}
If $|A|=3$ and $3$ divides $|I|-1$, then $NP$ is $\mathbf{K}$-manipulable and $\Omega^{NP}$-$\mathbf{K}$-strategy-proof. If $|A|=3$ and $3$ does not divide $|I|-1$, then $NP$ is $\Omega^{NP}$-$\mathbf{K}$-manipulable.
If $|A|\ge  4$ and $|I|< |A|-1$, then $NP$ is $\mathbf{K}$-strategy-proof.
If $|A|\ge  4$ and $|I|\ge |A|-1 $, then $NP$ is $\mathbf{K}$-manipulable and $\Omega^{NP}$-$\mathbf{K}$-strategy-proof. 
\end{theorem}

The aforementioned theorems state that, if the alternatives are at least three, some important differences among the three {\sc scc}s start emerging. Indeed, $BO$ is always $\Omega^{BO}$-$\mathbf{K}$-manipulable; $PL$ is $\Omega^{PL}$-$\mathbf{K}$-manipulable  unless the number of individuals is two or three; $NP$ exhibits instead a much more complex behavior. It is worth noting that, depending on the arithmetical relation between $|A|$ and $|I|$, all the three possible scenarios can occur for $NP$. In particular, if $|I|\ge |A|-1$ and the individuals have a level of information corresponding at most to the knowledge of the winners, then no individual can manipulate $NP$ in the sense of Kelly.
On the other hand, as is easily checked, if $|A|\geq 4$ and $|I|< |A|-1$, $NP$ is not set-monotonic. Thus there are examples beyond Theorem 1 in Brant (2015) for which $\mathbf{K}$-strategy proofness is possible.

The proofs of Theorems \ref{main-borda}, \ref{main-plurality} and \ref{main-neg-plurality} are given in Appendixes \ref{appendix-B} and \ref{appendix-C}.

\subsection{The Copeland social choice correspondence}

First, observe that in the case with two alternatives, the Copeland {\sc scc} is $\mathbf{K}$-strategy-proof and then $\Omega^F$-$\mathbf{K}$-strategy-proof. This follows from Proposition \ref{qualified} and the fact that the Copeland  {\sc scc} coincides with the simple majority rule when the alternatives are two. Restricting our attention to the case with at least three alternatives, we get the following result, the proof of which can be found in Appendix \ref{appendix-COP}.

\begin{theorem}\label{Copeland-main}
If $|A|=3$ and $|I|$ is odd, then $CO$ is $\mathbf{K}$-strategy-proof. 
If $|A|=3$ and $|I|$ is even, then $CO$ is $\Omega^{CO}$-$\mathbf{K}$-manipulable.
If $|A|\ge 4$, then $CO$ is $\Omega^{CO}$-$\mathbf{K}$-manipulable. 
\end{theorem}

\section{Conclusion}

After having introduced the definition of $\Omega$-$\mathbf{E}$-manipulability (Definition \ref{defpimanip}), we have focused on 
$\Omega^F$-$\mathbf{K}$-manipulability. We have then provided sufficient conditions for $\Omega^F$-$\mathbf{K}$-manipulability (Theorem \ref{general}) and analyzed $\Omega^F$-$\mathbf{K}$-manipulability when $F$ is positional (Theorems \ref{main-positional}, \ref{main-borda}, \ref{main-plurality}, and \ref{main-neg-plurality}) and when $F$ is the Copeland {\sc scc} (Theorem \ref{Copeland-main}). The results for $BO$ and $PL$ confirm that, even if we consider conditions that are favorable for positive results by using the Kelly extension rule and by reducing the information at the disposal of individuals, it is very difficult to avoid manipulation for those {\sc scc}s. On the other hand, the results for $NP$ show that, for this {\sc scc}, limiting the information can be an effective way to achieve strategy-proofness. 

This study can be extended and deepened by considering different extension rules, different information function profiles, and different families of social choice correspondences. Since different information function profiles may be comparable, meaning that one may be more or less informative than another, finding different results in terms of manipulability for different information function profiles might allow us to shed light on how much information is needed to manipulate a given {\sc scc}. By varying the extension rules, we modify our assumptions about how individuals evaluate sets of alternatives based on their preferences on alternatives. Of course, changing extension rules and information function profiles generally has a decisive impact on  manipulabilityresults.

As an example, consider the well-known Fishburn extension rule (G\"anderfors, 1976) defined, for every $q \in \mathcal{L}(A)$, by
\[
\mathbf{F}(q) \coloneq \left\{(B,C)\in P_0(A)^2 : x \succeq_q y\ \hbox{ for all}\  x \in B\setminus C\mbox{ and } y\in C \right\}
\]
\[ 
\cap \left\{(B,C)\in P_0(A)^2 : x \succeq_q y\mbox{ for all } x \in B\mbox{ and } y \in C\setminus B\right\}.
\]
Then the following result holds true.

\begin{proposition}\label{finale}
If $|A|\ge 3$ and $|A|$ does not divide $|I|-1$, then $NP$ is $\Omega^{NP}$-$\mathbf{F}$-manipulable.
\end{proposition}

\begin{proof} 
Let $|A|=n$ and $|I|=m$ and assume that $n\ge 3$, $m\ge 2$, and $n$ does not divide $m-1$. Assume also that $A=\{x_1,\ldots,x_n\}$, where $x_1,\ldots,x_n$ are distinct. Let $i \in I$, $q = [x_1,\ldots,x_n]$ and $\omega = \{\overline{p} \in \mathcal{L}(A)^{I \setminus \{i\}}: NP(\overline{p},q[i]) = A \setminus \{x_n\}\}$.
Note that $\omega\in\Omega_i^F(q)$. Indeed, let $\overline{p}' \in \mathcal{L}(A)^{I \setminus \{i\}}$ be such that $\overline{p}'(i)=q$ for all $i\in I \setminus \{i\}$. We have that $\overline{p}' \in\omega\neq\varnothing$ and so $\omega\in\Omega_i^F(q)$.

Let $\psi$ be the transposition that exchanges $x_n$ and $x_{n-1}$.
We show that $NP$ is $\Omega^{NP}$-$\mathbf{F}$-manipulable by proving that, for every $\overline{p}\in\omega$,
$NP(\overline{p},\psi q[i]) = A \setminus \{x_{n-1},x_{n}\}$. Indeed, since $N(\overline{p},q[i])=A\setminus \{x_n\}$, we deduce that, for every $\overline{p}\in\omega$, $NP(\overline{p},\psi q[i])\succ_{\mathbf{F}(q)} N(\overline{p},q[i])$, which clearly implies $\Omega^{NP}$-$\mathbf{F}$-manipulability.

Fix $\overline{p}\in\omega$ and set $p=(\overline{p},q[i])$.
For every $x \in A$, set $N(x) \coloneq \mathrm{np}(x,p)$,  $N'(x)\coloneq \mathrm{np}(x,(\overline{p},\psi q[i]))$, and $L(x)\coloneq |\{i\in I:\mathrm{rank}_{p(i)}(x)=n\}|$. We have that, for every $x\in A$, $L(x) = m - N(x)$. Moreover, $\sum_{x \in A} L(x) = m$ and, for every $x,y \in A \setminus \{x_n\}$, $L(x_n) > L(x) = L(y)$. Assume now by contradiction that $L(x_n) = L(x_{n-1}) + 1$. Then, we get 
\[
m = \sum_{x \in A} L(x) =L(x_n)+\sum_{x \in A \setminus \{x_n\}} L(x) = L(x_{n-1}) +1+(n-1) L(x_{n-1}) = n L(x_{n-1}) + 1.
\]
Thus $n$ divides $m-1$, a contradiction. Hence, $L(x_n) \ge L(x_{n-1}) +2$ and so $N(x_n) \le N(x_{n-1})-2$. 
We have $N'(x_n) = N(x_n) +1$, $N'(x_{n-1}) = N(x_{n-1}) - 1$ and, for every $x \in A \setminus \{x_{n-1},x_n\}$, $N'(x) = N(x)$. 
Note that, since $n\ge 3$, $ A \setminus \{x_{n-1},x_n\}\neq \varnothing$.
Then, $NP(\overline{p},\psi q[i]) = A \setminus \{x_{n-1},x_{n}\}$, as desired.
\end{proof}

Proposition \ref{finale} shows that, the use of the Fishburn extension rule instead of Kelly extension rule, makes $NP$ have a different behavior in terms of manipulability.  Indeed, if $|A|\ge 4$ and $|A|$ does not divide $|I|-1$, then, by Theorem \ref{main-neg-plurality} and Proposition \ref{finale}, we have that $NP$ is $\Omega^{NP}$-$\mathbf{K}$-strategy-proof and $\Omega^{NP}$-$\mathbf{F}$-manipulable.

\vspace{3mm}

\subsubsection*{Declarations of competing interest}

\vspace{-2mm}

The authors have no competing interests to declare that are relevant to the content of this article.

\subsubsection*{Data availability}

\vspace{-2mm}

No data was used for the research described in the article.

\subsubsection*{Declaration of Generative AI and AI-assisted technologies in the writing process}

\vspace{-2mm}

During the preparation of this work the authors used ChatGPT (chatgpt.com) for language refinement aimed at improving the clarity and readability of the manuscript. 
After using this tool, the authors reviewed and edited the content as needed and take full responsibility for the content of the article.

\subsubsection*{Acknowledgements}

\vspace{-2mm}

The authors thank the participants of the XLVIII AMASES Conference, for their comments. Raffaele Berzi and Michele Gori have been supported by local funding from the Universit\`a degli Studi di Firenze. Daniela Bubboloni has been supported by GNSAGA of INdAM (Italy), by the European Union - Next Generation EU, Missione 4 Componente 1, PRIN 2022-2022PSTWLB - Group Theory and Applications, CUP B53D23009410006, and by local funding from the Universit\`a degli Studi di Firenze.

\section*{Appendix}
\renewcommand{\thesubsection}{\Alph{subsection}}
\setcounter{subsection}{0}

This appendix is devoted to the proofs of Theorems \ref{main-borda}, \ref{main-plurality},  \ref{main-neg-plurality} and \ref{main-positional}. Those proofs are developed through some intermediate steps. For simplicity, in the rest of the paper we set $|A|=n$ and $|I|=m$ and, without loss of generality, we assume that $A=\ldbrack n \rdbrack$ and $I=\ldbrack m \rdbrack$.

\subsection{Properties of the Kelly extension rule}\label{appendix-A}

Recall that, for every $q \in \mathcal{L}(A)$, $\mathbf{K}(q)$ is a partial order. 
The following proposition collects some further basic properties of the Kelly extension rule. 

\begin{proposition}\label{properties} Let $q \in \mathcal{L}(A)$ with $q = [x_1,\ldots,x_{|A|}]$. Then, for every $B,C \in P_0(A)$, the following facts hold true.
\begin{itemize}
\item[$(i)$] If $B\succ_{\mathbf{K}(q)} C$, then $|B\cap C|\leq 1$.
\item[$(ii)$] If  $B\subsetneq C$ and $B$ and $C$ are $\mathbf{K}(q)$-comparable, then $B$ is a singleton.
\item[$(iii)$] $C\not\succ_{\mathbf{K}(q)}B$ if and only if $B=C$ or there exist $x\in B$ and $y\in C$ such that $x\succ_q y$.
\item[$(iv)$] If $\mathrm{rank}_q(y)>\mathrm{rank}_q(x)$, then $\{x\}\succ_{\mathbf{K}(q)} \{x,y\}\succ_{\mathbf{K}(q)} \{y\}$.
\item[$(v)$] If $B \neq \{x_1\}$, then $\{x_1\}\succ_{\mathbf{K}(q)} B$.
\item[$(vi)$] If $x_1 \in B$ and $x_1 \notin C$, then $C \nsucc_{\mathbf{K}(q)} B$.
\end{itemize}
\end{proposition}
\begin{proof} $(i)$ Let $B\succ_{\mathbf{K}(q)} C$. Suppose by contradiction that $|B\cap C|\ge 2$. Thus, there are $x,y\in B\cap C$ with $x\neq y$. Since $B\succeq_{\mathbf{K}(q)} C$ and $B\neq C$, we deduce that $x\succeq_q y$ and $y\succeq_q x$. By antisymmetry of $q$, we then get the contradiction $x=y$.

$(ii)$ Let $B\subsetneq C$ and suppose that $B\succeq_{\mathbf{K}(q)} C$ or $C\succeq_{\mathbf{K}(q)} B$. Since $B\neq C$, we have $B\succ_{\mathbf{K}(q)} C$ or $C\succ_{\mathbf{K}(q)} B$. Thus, by $(i)$, we get $|B\cap C|\leq 1$. Since $B\cap C=B$, we have $B\cap C\neq\varnothing$. As a consequence, we deduce $|B|=|B\cap C|=1$.

$(iii)$ Let $C\not\succ_{\mathbf{K}(q)}B$ and $B\neq C$.  Assume, by contradiction, that, for every $x\in B$ and $y\in C$, $x\not\succ_q y$. Since $q$ is complete, we deduce that, for every $x\in B$ and $y\in C$, $y\succeq_q x$. Thus, $C\succeq_{\mathbf{K}(q)} B$. Since $\mathbf{K}(q)$ is antisymmetric and $B\neq C$, we finally get the contradiction $C\succ_{\mathbf{K}(q)}B$.

Assume, conversely, that $B=C$ or that there exist $x\in B$ and $y\in C$ such that $x\succ_q y$. If $B=C$, we immediately have $C\not\succ_{\mathbf{K}(q)}B$. If $B\neq C$, there must exist $x\in B$ and $y\in C$ such that $x\succ_q y$. Then, $y\not\succeq_q x$ and therefore  $C\not \succeq_{\mathbf{K}(q)} B$. As a consequence, we also have $C\not\succ_{\mathbf{K}(q)}B$.

$(iv)$-$(vi)$ Straightforward.
\end{proof}

\subsection{The case \texorpdfstring{$|A|=2$}{|A|=2}}\label{sec-A2}

Assume that $A=\{a,b\}$ with $a\neq b$. Let $\alpha\ge \frac{|I|}{2}$. The $\alpha$-majority {\sc scc}, here denoted by $MAJ_\alpha$, is defined, for every $p\in\mathcal{L}(A)^I$, by
\[
MAJ_\alpha(p)\coloneq\left\{
\begin{array}{llll}
\{a\}&\mbox{ if }|\{i\in I: a\succeq_{p(i)}b\}|> \alpha\\
\{b\}&\mbox{ if }|\{i\in I: b\succeq_{p(i)}a\}|> \alpha\\
\{a,b\}&\mbox{ otherwise}\\
\end{array}
\right.
\]
It is easily observed that each positional {\sc scc} on two alternatives coincides with $MAJ\coloneq MAJ_\frac{|I|}{2}$. 

\begin{proposition}\label{qualified}
Assume that $n=2$ and $\alpha\ge \frac{m}{2}$. Then $MAJ_\alpha$ is $\mathbf{K}$-strategy-proof.
\end{proposition}

\begin{proof}
Let $i \in I$, $\overline{p} \in \mathcal{L}(A)^{I\setminus\{i\}}$ and $q,q'\in\mathcal{L}(A)$. Let us set $p=(\overline{p},q[i])$ and $p'=(\overline{p},q'[i])$ and assume that $q=[x_1,x_2]$, where $A=\{x_1,x_2\}$. We want to prove that
\begin{equation}\label{maj}
MAJ_\alpha(p')\nsucc_{\mathbf{K}(q)} MAJ_\alpha(p).
\end{equation}
If $q'=q$, \eqref{maj} is true. Assume $q'\neq q$ and so $q'=[x_2,x_1]$. If $MAJ_\alpha(p) = \left\{x_1\right\}$, we have that $MAJ_\alpha(p)\succeq_{\mathbf{K}(q)} B$ for all $B\in P_0(A)$ and thus \eqref{maj} holds true. Assume next that $MAJ_\alpha(p) \neq \left\{x_1\right\}.$ Thus, $x_2\in MAJ_\alpha(p)$ and $|\{i\in I: x_1\succeq_{p(i)}x_2\}|\leq \alpha$. Then, we have that $|\{i\in I: x_1\succeq_{p'(i)}x_2\}|< |\{i\in I: x_1\succeq_{p(i)}x_2\}|\leq \alpha$. As a consequence, $x_2\in MAJ_\alpha(p')$. If $x_1\in MAJ_\alpha(p)$, then \eqref{maj} is true. If $x_1\notin MAJ_\alpha(p)$, then $MAJ_\alpha(p)=\{x_2\}$. Thus, $|\{i\in I: x_2\succeq_{p'(i)}x_1\}|>|\{i\in I: x_2\succeq_{p(i)}x_1\}|> \alpha$ and hence $MAJ_\alpha(p')=\{x_2\}$. As a consequence, \eqref{maj} is true.
\end{proof}

\subsection{Proofs of Theorems  \ref{main-positional}, \ref{main-borda}, and  \ref{main-plurality}}\label{appendix-B}

\begin{proposition}\label{pos-wum}
Let $F$ be a positional {\sc scc}. Then $F$ satisfies {\sc wsm}.
\end{proposition}

\begin{proof}
Assume that $F$ is the positional {\sc scc} with scoring vector $w$.
Consider $i\in I$, $\overline{p}\in \mathcal{L}(A)^{I\setminus\{i\}}$, $q\in\mathcal{L}(A)$ and $x,y,z\in A$ such that 
\begin{itemize}
\item $F(\overline{p},q[i])=\{z\}$,  
\item $\mathrm{rank}_{q}(x)+1=\mathrm{rank}_{q}(y)<\mathrm{rank}_{q}(z)$,
\end{itemize}
and let $\psi\in \mathrm{Sym}(A)$ be the transposition that exchanges $x$ and $y$. 

We know that, for every $u\in A\setminus\{z\}$, $\mathrm{sc}_w(z,(\overline{p},q[i]))>\mathrm{sc}_w(u,(\overline{p},q[i]))$. It is immediately observed that $\mathrm{sc}_w(u,(\overline{p},\psi q[i]))=\mathrm{sc}_w(u,(\overline{p},q[i]))$ for all $u\in A\setminus\{x,y\}$. In particular, for every $u\in A\setminus\{x,y,z\}$,
we have $$\mathrm{sc}_w(z,(\overline{p},\psi q[i]))=\mathrm{sc}_w(z,(\overline{p},q[i]))>\mathrm{sc}_w(u,(\overline{p},q[i]))=\mathrm{sc}_w(u,(\overline{p},\psi q[i])).$$
Moreover, we have
 $$\mathrm{sc}_w(x,(\overline{p},\psi q[i]))\leq\mathrm{sc}_w(x,(\overline{p},q[i]))<\mathrm{sc}_w(z,(\overline{p},q[i]))=\mathrm{sc}_w(z,(\overline{p},\psi q[i])).$$
As a consequence, $F(\overline{p},\psi q[i])$ is a nonempty subset of $\{y,z\}$. 
\end{proof}

\begin{proposition}\label{us-unanimous}
Assume that $n\ge 3$ and $m\ge 4$, $m\neq 5$. If $F$ is a unanimous positional {\sc scc}, then $F$ is {\sc us}.   
\end{proposition}

\begin{proof}
Let $F$ be a unanimous positional {\sc scc} with scoring vector $w$, where  $w_1>w_2$.
We divide the proof into two different cases. 

    Assume first that $m$ is even. Let $\overline{p}\in \mathcal{L}(A)^{I\setminus\{m\}}$ be defined as follows: 
    \begin{itemize}
        \item for every $j\in \ldbrack\frac{m}{2}-1\rdbrack$,  $\overline{p}(j)\coloneq[1,2,3,(4),\dots,(n)]$;
        \item for every $j\in\{\frac{m}{2},\ldots, m-2\}$, $\overline{p}(j)\coloneq[2,1,3,(4),\dots,(n)]$;
        \item $\overline{p}(m-1)\coloneq[2,3,1,(4),\dots,n]$.
    \end{itemize}
     Consider then $q\coloneq[3,1,2,(4),\dots,(n)]$ and, for every $x\in A$, set $S(x)\coloneq\mathrm{sc}_w(x,(\overline{p},q[m]))$.
		We have that 
		\begin{eqnarray*}
	S(1)&=&\frac{m-2}{2}w_1+\frac{m-2}{2}w_2+w_2+w_3\\
	S(2)&=&\frac{m-2}{2}w_1+\frac{m-2}{2}w_2+w_1+w_3\\
	S(3)&=&(m-2)w_3+w_2+w_1
		\end{eqnarray*}
		and, for every $u\in A\setminus\{1,2,3\}$, $S(u)\le mw_3$.
		Recalling that $w_1>w_2$ and $m\ge 4$, we deduce that $F(\overline{p},q[m])=\{2\}$. Note also that
		$\mathrm{rank}_{q}(3)+1=\mathrm{rank}_{q}(1)<\mathrm{rank}_{q}(2)$.
		A computation finally shows that $F(\overline{p}, \psi q[m])=\{1,2\}$, where $\psi$ is the transposition that exchanges $1$ and $3$. Hence, $F$ satisfies {\sc us}.
		
	 	Assume now that $m$ is  odd. Then we have $m\ge 7$. Let $\overline{p}\in \mathcal{L}(A)^{I\setminus\{m\}}$ be defined as follows: 
        \begin{itemize}
            \item for every $j\in \ldbrack\frac{m-3}{2}\rdbrack$, $\overline{p}(j)\coloneq[1,2,3,(4),\dots,(n)]$;
            \item for every $j\in\{\frac{m-3}{2}+1,\ldots, m-3\}$, $\overline{p}(j)\coloneq[2,1,3,(4),\dots,(n)]$;
            \item $\overline{p}(m-2)\coloneq[3,2,1,(4),\dots,(n)]$;
            \item $\overline{p}(m-1)\coloneq[2,1,3,(4),\dots,(n)]$. 
        \end{itemize}
           Consider then $q\coloneq[3,1,2,(4),\dots,(n)]$ and, for every $x\in A$, set $S(x)\coloneq\mathrm{sc}_w(x,(\overline{p},q[m]))$.
		We have that 
		\begin{eqnarray*}
		S(1)&=&\frac{m-3}{2}w_1+\frac{m-3}{2}w_2+w_3+2w_2\\
		S(2)&=&\frac{m-3}{2}w_1+\frac{m-3}{2}w_2+w_1+w_2+w_3\\
		S(3)&=&(m-2)w_3+2w_1
		\end{eqnarray*}
		and, for every $u\in A\setminus\{1,2,3\}$, $S(u)\le m w_3$.
		Recalling that $w_1>w_2$ and $m\ge 7$, we deduce $F(\overline{p},q[m])=\{2\}$. Note also that
		$\mathrm{rank}_{q}(3)+1=\mathrm{rank}_{q}(1)<\mathrm{rank}_{q}(2)$.
		A computation finally shows that $F(\overline{p}, \psi q[m])=\{1,2\}$, where $\psi$ is the transposition that exchanges $1$ and $3$. Hence, $F$ satisfies {\sc us}.
\end{proof}

\begin{corollary}\label{corollary-pos}
Assume that $n\ge 3$, $m\ge 4$ and $m\neq 5$. If $F$ is a unanimous positional {\sc scc}, then $F$ is $\Omega^{F}$-$\mathbf{K}$-manipulable.  
\end{corollary}

\begin{proof}
Apply Theorem \ref{general} and Propositions \ref{pos-wum} and \ref{us-unanimous}.
\end{proof}

\begin{proposition}\label{main-positional-5}
Assume that $n\ge 3$ and $m=5$. If $F$ is a unanimous positional {\sc scc}, then $F$ is $\Omega^{F}$-$\mathbf{K}$-manipulable.  
\end{proposition}

        \begin{proof}
        Let $F$ be a unanimous positional {\sc scc}. Then, there exists a scoring vector $w$ such that $F=PS_w$ and $w_1>w_2$.
        
	Assume first that $w_2 > w_3$. In such a case, we can use the same argument of the second part of the proof of Proposition \ref{us-unanimous} to obtain that $F$ satisfies {\sc us}. As a consequence, by Theorem \ref{general} and Propositions \ref{pos-wum}, we get that $F$ is $\Omega^{F}$-$\mathbf{K}$-manipulable.
    
    Assume now that $w_2 = w_3$. Consider $q\coloneq[3,1,2,(4),\dots,(n)]$, $q'\coloneq[1,3,2,(4),\dots,(n)]$ and $\omega\coloneq\{\overline{p}\in \mathcal{L}(A)^{I\setminus\{5\}}: F(\overline{p},q[5])=\{1,2\} \}$. First, let us prove that $\omega\in \Omega^F_{5}(q)$. Indeed, let $\overline{p}'\in \mathcal{L}(A)^{I\setminus\{5\}}$ be such that $\overline{p}'(1)\coloneq[1,2,3,(4),\dots,(n)]$,  $\overline{p}'(2)\coloneq[2,1,3,(4),\dots,(n)]$, $\overline{p}'(3)\coloneq[1,2,3,(4),\dots,(n)]$, $\overline{p}'(4)\coloneq[2,1,3,(4),\dots,(n)]$. Moreover, set, for every $x\in A$, $S'(x)\coloneq\mathrm{sc}_w(x,(\overline{p}',q[5]))$. Recalling that $w_2=w_3$, we have that $S'(1)=2w_1+3w_2$, $S'(2)=2w_1+3w_2$, $S'(3)=w_1+4w_2$, and, for every $u\in A\setminus\{1,2,3\}$, $S'(u)\le 5w_2$. Recalling also that $w_1>w_2$, we deduce $F(\overline{p}',q[5])=\{1,2\}$. Thus, $\overline{p}'\in\omega$, and so $\omega\neq\varnothing$ and $\omega\in \Omega^F_{5}(q)$. 
		
		Consider now any $\overline{p}\in\omega$. We have $\mathrm{sc}_w(1,(\overline{p},q[5]))=\mathrm{sc}_w(2,(\overline{p},q[5]))>\mathrm{sc}_w(u,(\overline{p},q[5]))$ for all $u\neq A\setminus \{1,2\}$. It easily follows that $F(\overline{p},q'[5])=\{1\}\succ_{\mathbf{K}(q)}\{1,2\}=F(\overline{p},q[5])$. Thus, we conclude that $F$ is $\Omega^{F}$-$\mathbf{K}$-manipulable. 
\end{proof}

\begin{proof}[Proof of Theorem \ref{main-positional}]
    Apply Corollary \ref{corollary-pos} and Proposition \ref{main-positional-5}.
\end{proof}

\begin{proof}[Proof of Theorem \ref{main-borda}]
If $n=2$, then apply Theorem \ref{qualified}. 
    
Assume now $n\ge 3$ and $m=2$. Let $q \coloneq [2,1,3,(4),\dots,(n)]$, $q'\coloneq[2,3,1,(4),\dots,(n)]$, $\omega \coloneq \{\overline{p} \in \mathcal{L}(A)^{I \setminus \{2\}} : F(\overline{p},q[2]) = \{1,2\}\}$, and $\overline{p}'\in\mathcal{L}(A)^{I \setminus \{2\}}$ be such that  $\overline{p}'(1) \coloneq[1,2,3,(4),\dots,(n)]$. We have $F(\overline{p}',q[2]) = \{1,2\}$, hence $\omega \neq \varnothing$ and $\omega \in \Omega^F_{2}(q)$. We show that $BO$ is $\Omega^{BO}$-$\mathbf{K}$-manipulable proving that, for every $\overline{p}\in \omega$, $BO(\overline{p},q'[2])\succ_{\mathbf{K}(q)} BO(\overline{p},q[2])$. Let $\overline{p}\in \omega$. Thus, $BO(\overline{p},q[2])=\{1,2\}$, and then $\mathrm{bo}(1,(\overline{p},q[2]))=\mathrm{bo}(2,(\overline{p},q[2]))>\mathrm{bo}(x,(\overline{p},q[2]))$ for all $x\in A \setminus \{1,2\}$. 
Suppose that $\mathrm{bo}(1,(\overline{p},q[2]))=\mathrm{bo}(3,(\overline{p},q[2]))+1$. We get 
\[
(n-\mathrm{rank}_{\overline{p}(1)}(1))+(n-2)=(n-\mathrm{rank}_{\overline{p}(1)}(3))+(n-3)+1,
\]
that is, $\mathrm{rank}_{\overline{p}(1)}(1)=\mathrm{rank}_{\overline{p}(1)}(3),$
a contradiction. Thus, we have $\mathrm{bo}(1,(\overline{p},q[2]))\ge \mathrm{bo}(3,(\overline{p},q[2]))+2$.
It is easily observed that  
$\mathrm{bo}(1,(\overline{p},q'[2]))=\mathrm{bo}(1,(\overline{p},q[2]))-1$;
$\mathrm{bo}(2,(\overline{p},q'[2]))=\mathrm{bo}(2,(\overline{p},q[2]))$;
$\mathrm{bo}(3,(\overline{p},q'[2]))=\mathrm{bo}(3,(\overline{p},q[2]))+1$;
for all $x\in A \setminus \{1,2,3\}$, $\mathrm{bo}(x,(\overline{p},q'[2]))=\mathrm{bo}(x,(\overline{p},q[2]))$.
That implies $BO(\overline{p},q'[2])=\{2\}$ and hence we conclude that $BO(\overline{p},q'[2])=\{2\}\succ_{\mathbf{K}(q)} 
\{1,2\}=BO(\overline{p},q[2])$.
    
Assume now $n\ge 3$ and $m=3$. We prove that $BO$ is {\sc us} and we complete the proof applying Proposition \ref{pos-wum} and Theorem \ref{general}. Let $\overline{p} \in \mathcal{L}(A)^{I \setminus \{3\}}$ be defined by $\overline{p}(1)=\overline{p}(2)\coloneq[3,1,2,(4),\dots,(n)]$ and
 $q\coloneq[2,1,3,(4),\dots,(n)]\in\mathcal{L}(A)$. Observe that 
$BO(\overline{p},q[3])=\{3\}$, $\mathrm{rank}_{q}(2)+1=\mathrm{rank}_{q}(1)<\mathrm{rank}_{q}(3)$,
and $BO(\overline{p},\psi q\left[3\right])=\{1,3\}$, where $\psi$ is the transposition that exchanges $1$ and $2$. Thus, $BO$ is {\sc us}.

Finally, if $n\ge 3$ and $m\ge 4$, then apply Theorem \ref{main-positional}. 
\end{proof}

\begin{proposition}\label{main-plu-23}
    Assume that $n\ge 3$ and $m\in \{2,3\}$. Then $PL$ is $\mathbf{K}$-strategy-proof.
\end{proposition}

\begin{proof} 

For every $q \in \mathcal{L}(A)$, we denote by $\mathrm{top}(q)$ the best alternative in $q$. 

Assume first that $I=\{1,2\}$. Let $i \in I$, $\overline{p}\in \mathcal{L}(A)^{I\setminus\{i\}}$, and $q \in \mathcal{L}(A)$. Denote by $j$ the unique element in $I\setminus\{i\}$. Of course, we have $\mathrm{top}(q) \in PL(\overline{p},q[i])$. We split the proof in two cases. 
	\begin{itemize}
		\item If $\mathrm{top}(p(j)) = \mathrm{top}(q)$, then $PL(\overline{p},q[i]) = \{\mathrm{top}(q)\}$. In that case, for every $q' \in \mathcal{L}(A)$, we have $PL(\overline{p},q'[i])\nsucc_{\mathbf{K}(q)} PL(\overline{p},q[i])$.
            \item If $\mathrm{top}(p(j)) \neq \mathrm{top}(q)$, then we have $PL(p) = \{\mathrm{top}(p(j)),\mathrm{top}(q)\}$. Consider $q' \in \mathcal{L}(A)$. If $\mathrm{top}(q') = \mathrm{top}(q)$, we have $PL(\overline{p},q'[i]) = PL(\overline{p},q[i])$.
            If $\mathrm{top}(q') \neq \mathrm{top}(q)$, we have $PL(\overline{p},q'[i]) = \{\mathrm{top}(p(j)),\mathrm{top}(q')\}$.
		In both cases, we have $PL(\overline{p},q'[i]) \nsucc_{\mathbf{K}(q)} PL(\overline{p},q[i])$. 
	\end{itemize} 
Assume then that $I=\{1,2,3\}$. We first observe  that, for every $p \in \mathcal{L}(A)^I$, we have that $|PL(p)| = 1$ or $|PL(p)| = 3$. Indeed, assume first that, for every distinct $i,j \in I$, we have $\mathrm{top}(p(i)) \neq \mathrm{top}(p(j))$. Then $|PL(p)| = 3$. Assume next that there exist $x \in A$ and $i,j \in I$ distinct such that $x = \mathrm{top}(p(i)) = \mathrm{top}(p(j))$. Then $PL(p) = \{x\}$. Consider now $i \in I$, $\overline{p}\in \mathcal{L}(A)^{I\setminus\{i\}}$, and $q \in \mathcal{L}(A)$. 
	\begin{itemize}
		\item If $|PL(\overline{p},q[i])| = 3$, then we have $\mathrm{top}(q) \in PL(\overline{p},q[i])$. Let $q' \in \mathcal{L}(A)$. If $\mathrm{top}(q') = \mathrm{top}(q)$, then $PL(\overline{p},q'[i])=PL(\overline{p},q[i])$ and that implies $PL(\overline{p},q'[i]) \nsucc_{\mathbf{K}(q)} PL(\overline{p},q[i])$. If instead $\mathrm{top}(q') \neq \mathrm{top}(q)$, then  we have $\mathrm{top}(q) \not\in PL(\overline{p},q'[i])$. Since $\mathrm{top}(q) \in PL(\overline{p},q[i])$, we have $PL(\overline{p},q'[i]) \nsucc_{\mathbf{K}(q)} PL(\overline{p},q[i])$.
		\item If $|PL(\overline{p},q[i])| = 1$, then we have two possibilities: if $PL(\overline{p},q[i]) = \{\mathrm{top}(q)\}$, then, for every $C \in \mathcal{P}_0(A)$, we have that $C \nsucc_{\mathbf{K}(q)} \{\mathrm{top}(q)\}$. We deduce that, for every $q' \in \mathcal{L}(A)$, we have 
		$PL(\overline{p},q'[i])\nsucc_{\mathbf{K}(q)} PL(\overline{p},q[i])$. If $PL(\overline{p},q[i]) = \{x\}$, with $x\neq \mathrm{top}(q)$ then, for both the individuals in $I\setminus\{i\}$, $x$ is the best alternative. As a consequence, for any $q'$, we have  $\{x\}= PL(\overline{p},q'[i])$ and thus $PL(\overline{p},q'[i])\nsucc_{\mathbf{K}(q)} PL(\overline{p},q[i])$. 
	\end{itemize}
	That proves that $PL$ is $\mathbf{K}$-strategy-proof. 
\end{proof}

\begin{proof}[Proof of Theorem \ref{main-plurality}]
    Apply Theorem \ref{main-positional} and Proposition \ref{main-plu-23}.
\end{proof}

\subsection{Proof of Theorem \ref{main-neg-plurality}}\label{appendix-C}

 \begin{lemma} \label{lemappres}
 The following facts hold:
 \begin{itemize}
 \item[$(i)$] For every $p\in \mathcal{L}(A)^I,$ we have $|NP(p)|\geq \max\{1, n-m\}$;
 \item[$(ii)$] $|NP(p)|\geq 2$ for all $p\in \mathcal{L}(A)^I$  if and only if $n-m\geq 2$.
 \end{itemize}
 \end{lemma}
\begin{proof} For every $p\in \mathcal{L}(A)^I$, define $Z(p)\coloneq\{z\in A: \exists i\in I \hbox{ such that } \mathrm{rank}_{p(i)}( z)=n\}.$

$(i)$ Clearly we have $|Z(p)|\leq \min\{n, m\}.$ Now, observe that $A\setminus Z(p)\subseteq NP(p)$ and thus 
\[
|NP(p)|\geq |A|-|Z(p)|\geq n-\min\{n, m\}=\left\{
\begin{array}{llll}
0 &\mbox{ if }n\leq m\\
n-m &\mbox{ if }n>m\\
\end{array}
\right.=\max\{0,n-m\}
\]
Since $NP(p)\neq \varnothing$, we deduce $|NP(p)|\geq \max\{1, n-m\}$.

$(ii)$ If $n-m\geq 2$, by $(i)$, we have $|NP(p)|\geq  n-m\geq 2$ for all $p\in \mathcal{L}(A)^I$. Assume next that $n-m\le 1$, that is, $m\geq n-1$.
Then, there exists $p\in \mathcal{L}(A)^I$ such that, for every $i\in I$,  $\mathrm{rank}_{p(i)}(n)=1$ and  $Z(p)=\ldbrack n-1 \rdbrack$. 
Of course, $\mathrm{np}(n,p)=m$ and $\mathrm{np}(x,p)\leq m-1$ for all $x\in \ldbrack n-1 \rdbrack$. As a consequence, $NP(p)=\{n\}$.
\end{proof}

\begin{proposition} \label{lemappsp}
Assume that $n-m\geq 2$. Then $NP$ is $\mathbf{K}$-strategy-proof.
\end{proposition}

\begin{proof}Let  $i \in I$, $q,q' \in \mathcal{L}(A)$ and $\overline{p} \in \mathcal{L}(A)^{I\setminus\{i\}}$. We set $p\coloneq (\overline{p},q\left[i\right])$, $p'\coloneq (\overline{p},q'\left[i\right])$, $B\coloneq NP(p)$ and $B'\coloneq NP(p')$. Moreover, for every $y\in A$, we set $N(y)\coloneq \mathrm{np}(y,p)$ and $N'(y)\coloneq \mathrm{np}(y,p')$.
    
In order to show that  $NP$ is $\mathbf{K}$-strategy-proof, we need to show that   
\begin{equation}\label{NP-stra}
B'\not \succ_{\mathbf{K}(q)} B.
\end{equation}
If $|B\cap B'|\geq 2$, then, by Proposition \ref{properties}$(i)$, we immediately get \eqref{NP-stra}.
If instead $|B \cap B'|\in \{0,1\}$, then, by Lemma \ref{lemappres}$(ii)$, we know that $|B|\ge 2$ and $|B'|\geq 2$. As a consequence, $B\not\subseteq B'$ and $B'\not\subseteq B$. We divide the argument into the two cases $|B \cap B'| = 0$ and $|B \cap B'| = 1$.
	
Assume first that $|B \cap B'| = 0$. Suppose that there exists $x'\in B'$ such that $\mathrm{rank}_q(x')=n.$ Pick $x\in B$. Since $B\cap B'=\varnothing$, we deduce $x\succ_qx'.$ Thus, by Proposition \ref{properties}$(iii)$, we deduce \eqref{NP-stra}.
Suppose instead that, for every $x'\in B'$, $\mathrm{rank}_q(x')\leq n-1$.
Let $x'\in B'$ and $x\in B$. Since $B\cap B'=\varnothing$, we have that $x'\notin B$ and then $N(x')<N(x)$.
Thus, we deduce	$N'(x')\leq N(x')<N(x)$, and then $N'(x')\leq N(x)-1\leq N'(x)$. Since $x'\in B'$, we conclude that $x\in B'$, a contradiction.
		
Assume now that $|B \cap B'| = 1$ and let $ B \cap B'=\{x\}.$ Since $B,B'$ are not included one in the other and have size at least $2$, there exist $z\in B\setminus\{x\}$ and $z'\in B'\setminus\{x\}$. Then $z'\notin B$ and we have 
	\begin{equation}\label{inizio2}
		N(x)=N(z)>N(z').
	\end{equation}
	Moreover, $z\notin B'$ and we have 
	\begin{equation}\label{inizio3}
		N'(x)=N'(z')>N'(z).
	\end{equation}
	Assume first that $\mathrm{rank}_q(x)=n$ and $\mathrm{rank}_{q'}(x)\leq n-1.$ Then we have $N'(x)>N(x)$	and thus $N(z')\geq N'(z')=N'(x)>N(x)$,	against the fact that $x\in B$.
	Assume next that $\mathrm{rank}_q(x)\leq n-1$ and $\mathrm{rank}_{q'}(x)= n.$ 
	Then, we have $N(x)>N'(x)$ and thus $N'(z)\geq N(z)=N(x)>N'(x)$,
	against the fact $x\in B'$.
	Assume  now that we have both $\mathrm{rank}_q(x)\leq n-1$ and $\mathrm{rank}_{q'}(x)\leq n-1$, or both $\mathrm{rank}_q(x)=n$ and $\mathrm{rank}_{q'}(x)= n.$
	Then, we have $N(x)=N'(x)$.
	As a consequence, using \eqref{inizio3}, we get
	$N(z)=N(x)=N'(x)>N'(z)$, and so we deduce that $z$ is not the worst alternative for $q$; 
using also \eqref{inizio2}, we get $N'(z')=N'(x)=N(x)>N(z')$, and so we deduce that $z'$ is the worst alternative for $q$. We then conclude that  $z\succ_qz'$ and, by Proposition \ref{properties}$(iii)$, we finally obtain \eqref{NP-stra}.
\end{proof}

\begin{proposition} \label{propNPKM}
Assume that $m\geq n-1$. Then $NP$ is $\mathbf{K}$-manipulable.
\end{proposition}

\begin{proof} For every $p\in \mathcal{L}(A)^I$, we set $Z(p)\coloneq\{z\in A: \exists i\in I \hbox{ such that } \mathrm{rank}_{p(i)}( z)=n\}$. 
	In order to prove that $NP$ is $\mathbf{K}$-manipulable, we exhibit $p \in \mathcal{L}(A)^I$, $i \in I$ and $q' \in \mathcal{L}(A)$ such that
	\begin{equation}\label{caso1}
		NP(p_{|_{I \setminus \left\{i\right\}}},q'\left[i\right])\succ_{\mathbf{K}(p(i))} NP(p).
	\end{equation}
    Since $n\geq 3$ and $m\geq n-1$, we have $A\setminus \{1,2\}\neq \varnothing$ and $m\geq n-2$.
Thus, there exists $p \in \mathcal{L}(A)^I$ such that $Z(p)=A\setminus \{1,2\}$ and $1\succ_{p(i)}2$ for all $i\in I$. Of course,  $NP(p) = \{1,2\}$.
	Let $z^*\in Z(p)$ be an alternative that is the worst alternative for the maximum number of individuals according to $p$. Since $m> n-2$, we have that $z^*$ is the worst alternative for at least two individuals. 
	
	Let $i$ be one of the individuals that considers $z^*$ her worst alternative and let $\psi$ be the transposition that exchanges $2$ and $z^*$. Define $q'\coloneq \psi p(i)\in \mathcal{L}(A)$ and $p'\coloneq  (p_{|_{I \setminus \left\{i\right\}}},q'\left[i\right])$.
We have that $NP(p')=\{1\}$ since $\mathrm{np}(1,p')=m$, $\mathrm{np}(2,p')=m-1$, $\mathrm{np}(z^*,p')=\mathrm{np}(z^*,p)+1\leq (m-2)+1=m-1$, and $\mathrm{np}(x,p')=\mathrm{np}(x,p)\leq m-1$ for all $x\in A\setminus\{1,2,z^*\}$. By Proposition \ref{properties}$(iv)$, $1\succ_{p(i)}2$ implies 
	$\{1\}\succ_{\mathbf{K}(p(i))}\{1,2\}$, and hence  \eqref{caso1} is finally shown.
\end{proof}

\begin{lemma} \label{lemquasiGSP}
	Let $i \in I$, $q,q' \in \mathcal{L}(A)$ and $\omega \in \Omega^{NP}_i(q)$. Let $B\in P_0(A)$ be the set such that 
	$$\omega=\{\overline{p}\in \mathcal{L}(A)^{I\setminus\{i\}}: NP(\overline{p},q[i])=B\}.$$ Let $z$ be the worst alternative in $q$ and $z'$ be the worst alternative in $q'$. If  $z = z'$, or 
	 $z \in B$, or $|B| \le n-2$,
	then one of the following facts hold:
	\begin{itemize}
			\item[$(i)$] for every $\overline{p}\in\omega$, $NP(\overline{p},q'[i])\not\succ_{\mathbf{K}(q)}B$;
			\item[$(ii)$] there exists $\overline{p}'\in\omega$ such that $NP(\overline{p}',q'[i])\not\succeq_{\mathbf{K}(q)}B.$
	\end{itemize}
\end{lemma}

\begin{proof}
    Assume first that $z = z'$. We show that $(i)$ holds.
    Let $\overline{p}\in\omega$.  For every $x\in A$,  we have $\mathrm{np}(x,(\overline{p},q\left[i\right]))= \mathrm{np}(x,(\overline{p},q'\left[i\right]))$, thus $NP(\overline{p},q'\left[i\right]) =B\not\succ_{\mathbf{K}(q)}B$.
    
    Assume next that $z \neq z'$ and $z \in B$. We show that $(i)$ holds. 
    Let $\overline{p}\in\omega$.  Define $p \coloneq (\overline{p},q\left[i\right])$ and $p' \coloneq (\overline{p},q'\left[i\right])$. Since $z \in B = NP(p)$ we have that, for every $x \in A \setminus \{z\}$, 
	\begin{equation}\label{z maximize}
		\mathrm{np}(z,p) \geq \mathrm{np}(x,p).
	\end{equation}
	Moreover, since $z$ is not the worst alternative in $q'$, we have
	\begin{equation}\label{z not-worst}
		\mathrm{np}(z,p') =\mathrm{np}(z,p)+1
	\end{equation}
	and, for every $x \in A \setminus \{z\}$, 
	\begin{equation}\label{z not-worst2}
		\mathrm{np}(x,p') \leq \mathrm{np}(x,p).
	\end{equation}
    As a consequence, by \eqref{z not-worst},  \eqref{z maximize},  \eqref{z not-worst2}, we get $\mathrm{np}(z,p')>\mathrm{np}(z,p)\geq  \mathrm{np}(x,p)\geq \mathrm{np}(x,p')$ for all $x \in A \setminus \{z\}$, which gives $NP(p')=\{z\}\not\succ_{\mathbf{K}(q)}B$.
    
    Assume finally that $z \neq z'$, $z \notin B$ and $|B| \le n-2$. We show that $(ii)$ holds. Consider $C\coloneq  A \setminus B$. Thus, $z \in C$ and $|C| \ge 2$. For every $\overline{p}\in \mathcal{L}(A)^{I \setminus \left\{i\right\}}$, define $$Z(\overline{p})\coloneq\{w\in A: \exists j\in I\setminus \left\{i\right\} \hbox{ such that } \mathrm{rank}_{p(j)}( w)=n\}.$$
	By Lemma \ref{lemappres} $(i)$, we have $|B| \ge n-m$, hence $$|C \setminus \{z\}| = |C| - 1 = n-|B|-1 \le n+m-n-1=m-1.$$ Thus, we can construct $\overline{p}'\in \mathcal{L}(A)^{I \setminus \left\{i\right\}}$ such that $Z(\overline{p}')=C\setminus\{z\}.$ Clearly we have $NP(\overline{p}', q[i])=B$ and hence $\overline{p}'\in \omega.$
	Set now $p\coloneq (\overline{p}', q[i])$ and $p'\coloneq (\overline{p}', q'[i])$. 
	Recalling that $z$  is not the worst alternative in $q'$, we have $\mathrm{np}(z,p') =\mathrm{np}(z,p)+1=(m-1)+1=m.$
	As a consequence, $z\in NP(p')$. Since $z \notin B$ we have $NP(\overline{p}',q'[i])\not\succeq_{\mathbf{K}(q)}B.$
\end{proof}

\begin{proposition} \label{propNPomegaKSP}
Assume that $n$ divides $m-1$ or $n \ge 4$. Then $NP$ is $\Omega^{NP}$-$\mathbf{K}$-strategy-proof.
\end{proposition}

\begin{proof}
Let us consider $i \in I$, $q,q' \in \mathcal{L}(A)$ and $\omega \in \Omega_i^{NP}(q)$ and prove that one of the following facts hold: 
\begin{itemize}
\item[$(a)$] for every $\overline{p}\in\omega$, $NP(\overline{p},q'[i])\not\succ_{\mathbf{K}(q)}B$,
\item[$(b)$] there exists $\overline{p}'\in\omega$ such that $NP(\overline{p}',q'[i])\not\succeq_{\mathbf{K}(q)}B,$
\end{itemize}
where $B\in P_0(A)$ is such that $\omega=\{\overline{p}\in \mathcal{L}(A)^{I\setminus\{i\}}: NP(\overline{p},q[i])=B\}$.
Let $z$ be the worst alternative in $q$, and $z'$ be the worst alternative in $q'$. 

If $z = z'$ or $z \in B$ or $|B| \le n-2$, then, by Lemma \ref{lemquasiGSP}, we know that one between $(a)$ and $(b)$ holds.

Assume next that $z \neq z'$,  $z \notin B$ and $|B| \ge n-1$. Since $z \notin B$, we actually have $|B|= n-1$. Note also that, under these assumptions, $B = A \setminus \{z\}$. 
\begin{itemize}
\item Assume that $n$ divides $m-1$. We prove that $(b)$ holds.  Defining $c \coloneq \frac{m-1}{n}$, we can consider $\overline{p}'\in \mathcal{L}(A)^{I \setminus \left\{i\right\}}$ such that, for every $x \in A$, $x$ is the worst alternative of exactly $c$ individuals. Since $z$ is the worst alternative in $q$, we have $NP(\overline{p}',q[i]) = B$ and then $\overline{p}'\in \omega$.
Since $z \neq z'$, we have $NP(\overline{p}',q'[i]) = A\setminus \{z'\}$. Since $A\setminus \{z'\}\nsucceq_{\mathbf{K}(q)} B$, we conclude that $NP(\overline{p}',q'[i]) \nsucceq_{\mathbf{K}(q)} B$.
\item Assume that $n \ge 4$. We prove that $(b)$ holds. Consider $\overline{p}' \in \mathcal{L}(A)^{I \setminus \{i\}}$ such that $z$ is the worst alternative of all individuals.
Of course, $NP(\overline{p}',q[i]) = B$ and then $\overline{p}'\in \omega$. Since $z'\neq z$, we have that $NP(\overline{p}',q'[i]) = A \setminus \{z,z'\}\subsetneq B$. Note that $|A \setminus \{z,z'\}|=n-2\geq 2$. By Proposition \ref{properties}$(ii)$, we deduce that $B$ and $NP(p')$ are not $\mathbf{K}(q)$-comparable. It follows that $NP(\overline{p}',q'[i])\not\succeq_{\mathbf{K}(q)}B.$
\end{itemize}
\end{proof}
	
\begin{proposition} \label{prop3notdivideNP}
Assume that $n = 3$ and $3$ does not divide $m-1$. Then $NP$ is $\Omega^{NP}$-$\mathbf{K}$-manipulable;
\end{proposition}

\begin{proof}
    Let $i \in I$, $q\coloneq [1,2,3]$ and $\omega \coloneq\{\overline{p} \in \mathcal{L}(A)^{I \setminus \{i\}}: NP(\overline{p},q[i]) = \{1,2\}\}$. 
    Let $\overline{p}'\in \mathcal{L}(A)^{I \setminus \{i\}}$ be such that $3$ is the worst alternative of all individuals. Thus we have $\overline{p}'\in \omega$, and hence $\omega\neq\varnothing$ and $\omega\in\Omega_i^{NP}(q)$.
    
    Consider now $\overline{p} \in \omega$. Set $p \coloneq (\overline{p},q[i])$  and, for every $x \in A$, $N(x) \coloneq \mathrm{np}(x,p)$ and 
    \[
    L(x)\coloneq |\{i\in I:\mathrm{rank}_{p(i)}(x)=3\}|.
    \]
    Note that $L(x) = m - N(x)$ for all $x\in A$,  $\sum_{x \in A} L(x) = m$, and $L(1) = L(2)< L(3)$. Assume by contradiction $L(3) = L(1) + 1$. Then, we get 
    \[
    m = \sum_{x \in A} L(x) = L(1) + L(2) + L(3) = 2L(1) + L(1) + 1 = 3L(1) + 1.
    \]
    Thus $n = 3$ divides $m-1$, a contradiction. Hence, $L(3) \ge L(1) +2$ and $L(3) \ge L(2) +2$, and so $N(3) \le N(1) - 2$ and $N(3) \le N(1) - 2$. Let $\psi$ be the transposition that exchanges $3$ and $2$ and set, for every $x \in A$, $N'(x) \coloneq \mathrm{np}(x,(\overline{p},\psi q[i]))$. We have $N'(3) = N(3) +1$, $N'(2) = N(2) - 1$ and, $N'(1) = N(1)$. Then $NP(\overline{p},\psi q[i]) = \{1\} \succ_{\mathbf{K}(q)} \{1,2\}$. Thus, $NP$ is $\Omega^{NP}$-$\mathbf{K}$-manipulable.
\end{proof}

\begin{proof}[Proof of Theorem \ref{main-neg-plurality}]
    Apply Propositions \ref{lemappsp}, \ref{propNPKM}, \ref{propNPomegaKSP} and \ref{prop3notdivideNP}.
\end{proof}

\subsection{Proof of Theorem \ref{Copeland-main}}\label{appendix-COP}

\begin{lemma} \label{lemsumscore}
Let $p \in \mathcal{L}(A)^I$. Then $\sum_{x\in A} \mathrm{co}(x,p) = 0$.
\end{lemma}

\begin{proof} Since $\sum_{x\in A} w_p(x) = \sum_{x\in A} l_p(x)$, we immediately obtain that 
$\sum_{x\in A} \mathrm{co}(x,p) = \sum_{x\in A} (w_p(x) - l_p(x)) =0$.
\end{proof}

\begin{proposition} \label{34}
$CO$ satisfies {\sc wsm}.
\end{proposition}

\begin{proof}
    Consider $i\in I$, $\overline{p}\in \mathcal{L}(A)^{I\setminus\{i\}}$, $q\in\mathcal{L}(A)$ and $x,y,z \in A$ such that 
    \begin{itemize}
    \item $CO(\overline{p},q[i])=\{z\}$,  
    \item $\mathrm{rank}_{q}(x)+1=\mathrm{rank}_{q}(y)<\mathrm{rank}_{q}(z)$,
    \end{itemize}
    and let $\psi\in \mathrm{Sym}(A)$ be the transposition that exchanges $x$ and $y$. We must prove that $CO(\overline{p},\psi q[i]) \subseteq \{z,y\}$. Since $CO(\overline{p},q[i])=\{z\}$, we have that, for every $w \in A \setminus \{z\}$, $\mathrm{co}(z,(\overline{p},q[i])) > \mathrm{co}(w,(\overline{p},q[i]))$. Observe that  
 $\mathrm{co}(z,(\overline{p},\psi q[i])) = \mathrm{co}(z,(\overline{p},q[i]))$, $\mathrm{co}(x,(\overline{p},\psi q[i])) \le \mathrm{co}(x,(\overline{p},q[i]))$
    and, for every $w \in A \setminus \{x,y,z\}$, $\mathrm{co}(w,(\overline{p},\psi q[i])) = \mathrm{co}(w,(\overline{p},q[i]))$. Therefore, we deduce that, for every $w \in A \setminus \{y,z\}$, $\mathrm{co}(z,(\overline{p},\psi q[i])) > \mathrm{co}(w,(\overline{p},\psi q[i]))$. Thus, we conclude that $CO(\overline{p},\psi q[i]) \subseteq \{z,y\}$.
\end{proof}

\begin{proposition}\label{35}
    Assume that $n \ge 4$ and $m = 2$. Then, $CO$ is {\sc us}. 
\end{proposition}

\begin{proof}
    Since $m = 2$ we have $I = \{1,2\}$. Let $\overline{p} \in \mathcal{L}(A)^{I\setminus\{2\}}$ be such that $\overline{p}(1) \coloneq [1,2,3,4,(5),\dots,(n)]$ and $q \coloneq [4,3,1,2,(5),\dots,(n)]$. We have  
    $\mathrm{co}(1,(\overline{p}, q[2])) = n-3$, $\mathrm{co}(2,(\overline{p}, q[2])) = n-5$, $\mathrm{co}(3,(\overline{p}, q[2])) = n-4$, $\mathrm{co}(4,(\overline{p}, q[2])) = n-4$, and, for every $w \in A \setminus \{1,2,3,4\}$, $\mathrm{co}(w,(\overline{p}, q[2])) = n - 2w + 1$. Thus, $CO(\overline{p}, q[2]) = \{1\}$ and $\mathrm{rank}_q(4) + 1 = \mathrm{rank}_q(3) < \mathrm{rank}_q(1)$. Let now $\psi$ be the permutation that exchanges $3$ and $4$. We have, $\mathrm{co}(3,(\overline{p}, \psi q[2])) = \mathrm{co}(3,(\overline{p}, q[2])) + 1$, $\mathrm{co}(4,(\overline{p}, \psi q[2])) = \mathrm{co}(4,(\overline{p}, q[2])) - 1$, and, for every $w \in A \setminus \{3,4\}$, $\mathrm{co}(w,(\overline{p}, \psi q[2])) = \mathrm{co}(w,(\overline{p}, q[2]))$. Thus, $CO(\overline{p}, \psi q[2]) = \{1,3\}$. We then conclude that $CO$ is {\sc us}. 
\end{proof}

\begin{proposition}\label{36}
    Assume that $n \ge 3$ and $m \ge 4$ is even. Then, $CO$ is {\sc us}. 
\end{proposition}

\begin{proof}
    Let $\overline{p} \in \mathcal{L}(A)^{I\setminus\{m\}}$ be defined as follows: 
    \begin{itemize}
        \item for every $j \in \ldbrack\frac{m}{2}-1\rdbrack$, $\overline{p}(j) \coloneq [1,2,3,(4),\dots,(n)]$;
        \item $\overline{p}\left(\frac{m}{2}\right) \coloneq [1,3,2,(4),\dots,(n)]$;
        \item $\overline{p}\left(\frac{m}{2}+1\right) \coloneq [2,1,3,(4),\dots,(n)]$;
        \item for every $j \in \left\{l\in\mathbb{N}: \frac{m}{2}+2\le l\leq m-1\right\}$, $\overline{p}(j) \coloneq [3,2,1,(4),\dots,(n)]$.
    \end{itemize}
    Consider $q \coloneq [3,2,1,(4),\dots,(n)]$. We have $\mathrm{co}(1,(\overline{p}, q[m])) = n-2$, $\mathrm{co}(2,(\overline{p}, q[m])) = n-3$, $\mathrm{co}(3,(\overline{p}, q[m])) = n-4$, and, for every $w \in A\setminus\{1,2,3\}$, $\mathrm{co}(w,(\overline{p}, q[m])) = n-2w+1$. Thus, $CO(\overline{p}, q[m]) = \{1\}$ and $\mathrm{rank}_q(3) + 1 = \mathrm{rank}_q(2) < \mathrm{rank}_q(1)$. Let now $\psi$ be the permutation that exchanges $2$ and $3$. We have $\mathrm{co}(2,(\overline{p}, \psi q[m])) = \mathrm{co}(2,(\overline{p}, q[m])) + 1 = n-2$, $\mathrm{co}(3,(\overline{p}, \psi q[m])) = \mathrm{co}(3,(\overline{p}, q[m])) - 1 = n-5$, and,  for every $w \in A \setminus \{2,3\}$, $\mathrm{co}(w,(\overline{p}, \psi q[m])) = \mathrm{co}(w,(\overline{p}, q[m]))$. Thus, $CO(\overline{p}, \psi q[m]) = \{1,2\}$.  We then conclude that $CO$ is {\sc us}. 
\end{proof} 

\begin{proposition}\label{40}
Assume that $n \ge 5$ and $m$ is odd. Then, $CO$ is {\sc us}.
\end{proposition}

\begin{proof}
    Let $\overline{p} \in \mathcal{L}(A)^{I\setminus\{m\}}$ be defined as follows: 
    \begin{itemize}
        \item for every $j \in \ldbrack \frac{m-1}{2} \rdbrack$, $\overline{p}(j) \coloneq[1,2,3,4,5,(6)\dots,(n)]$;
        \item for every $j \in \{l\in\mathbb{N}:\frac{m-1}{2}+1\leq l\leq m-1\}$, $\overline{p}(j) \coloneq[3,4,1,5,2,(6),\dots,(n)]$.
    \end{itemize}
    Consider $q \coloneq [5,2,4,1,3,(6),\dots,(n)]$. We have
    $\mathrm{co}(1,(\overline{p}, q[m])) = n-3$, $\mathrm{co}(2,(\overline{p}, q[m])) = n-5$, $\mathrm{co}(3,(\overline{p}, q[m])) = n-5$, $\mathrm{co}(4,(\overline{p}, q[m])) = n-5$, $\mathrm{co}(5,(\overline{p}, q[m])) = n-7$, and,  for every $w \in A\setminus\{1,2,3,4,5\}$, $\mathrm{co}(w,(\overline{p}, q[m])) = n-2w+1$. Thus, $CO(\overline{p}, q[m]) = \{1\}$ and $\mathrm{rank}_q(2) + 1 = \mathrm{rank}_q(4) < \mathrm{rank}_q(1)$. Let now $\psi$ be the permutation that exchanges $2$ and $4$. We have $\mathrm{co}(4,(\overline{p}, \psi q[m])) = \mathrm{co}(4,(\overline{p}, q[m])) + 2 = n-3$, $\mathrm{co}(2,(\overline{p}, \psi q[m])) = \mathrm{co}(2,(\overline{p}, q[m])) - 2 = n-7$, and,  for every $w \in A \setminus \{2,4\}$, $\mathrm{co}(w,(\overline{p}, \psi q[m])) = \mathrm{co}(w,(\overline{p}, q[m]))$. Thus, $CO(\overline{p}, \psi q[m]) = \{1,4\}$.  We then conclude that $CO$ is {\sc us}. 
\end{proof}

\begin{proposition}\label{37}
Assume that $n = 3$ and $m = 2$. Then, $CO$ is $\Omega^{CO}$-$\mathbf{K}$-manipulable.
\end{proposition}

\begin{proof}
Let $i \in I$, and let $j$ be the other element of $I$. Thus, $ \mathcal{L}(A)^{I \setminus \{i\}}=\mathcal{L}(A)^{\{j\}}$. 
Let $q \coloneq [2,1,3]$ and $\omega \coloneq \left\{\overline{p} \in \mathcal{L}(A)^{\{j\}}: CO(\overline{p},q[i]) = \{1,2\}\right\}$.
By considering all the six elements in $\mathcal{L}(A)^{\{j\}}$, it is easy to show that 
$\omega = \{\overline{p}^* \}$, where $\overline{p}^* \in \mathcal{L}(A)^{\{j\}}$ is such that $\overline{p}^*(j)=[1,2,3]$. In particular, $\omega \neq \varnothing$ and $\omega \in \Omega^{CO}_i(q)$.
Let now set  $q' \coloneq [2,3,1]$. A simple calculation shows that $CO(\overline{p}^*,q'[i]) = \{2\}$, and $\{2\} \succ_{\mathbf{K}(q)} \{1,2\}$. Thus, $CO$ is $\Omega^{CO}$-$\mathbf{K}$-manipulable.
\end{proof}

\begin{proposition}\label{38}
Assume that $n = 3$ and $m$ is odd. Then, $CO$ is $\mathbf{K}$-strategy-proof.
\end{proposition}

\begin{proof}
    Since $m$ is odd, for $x,y \in A$ and $p \in \mathcal{L}(A)^I$ we cannot have $c_p(x,y) = c_p(y,x)$. Let $A = \{x,y,z\}$ and take any $p \in \mathcal{L}(A)^I$. The only possible values of $\mathrm{co}(x,p)$ are $2,0,-2$: if $c_p(x,y) > c_p(y,x)$ and $c_p(x,z) > c_p(z,x)$, then $\mathrm{co}(x,p) = 2$; if $c_p(x,y) > c_p(y,x)$ and $c_p(x,z) < c_p(z,x)$, then $\mathrm{co}(x,p) = 0$; if $c_p(x,y) < c_p(y,x)$ and $c_p(x,z) < c_p(z,x)$, then $\mathrm{co}(x,p) = -2$. By Lemma \ref{lemsumscore}, $\sum_{x\in A} \mathrm{co}(x,p) = 0$, which implies that we have $\mathrm{co}(x,p) = 2$ if and only if $CO(p) = \{x\}$. Moreover,  if $\mathrm{co}(x,p) = -2$, then $\mathrm{co}(y,p) = 2$ or $\mathrm{co}(z,p) = 2$, and the alternative with score $2$ is the unique winner. If otherwise we have $\mathrm{co}(x,p) = \mathrm{co}(y,p) = \mathrm{co}(z,p) = 0$, then $CO(p) = \{x,y,z\}$. Thus $|CO(p)| = 1$ or $|CO(p)| = 3$. Let now $i \in I$, $q,q' \in \mathcal{L}(A)$, and  suppose $q = [x,y,z]$. Take $\overline{p} \in \mathcal{L}(A)^{I\setminus \{i\}}$, set $p \coloneq (\overline{p},q[i])$ and $p'=(\overline{p},q'[i])$. We consider the two possibilities for the size of $CO(p)$ separately. 
    \begin{itemize}
        \item If $|CO(p)| = 1$, then we have three possibilities. If $CO(p) = \{x\}$, then there is no $B \subseteq A$ such that $B \succ_{\mathbf{K}(q)} CO(p)$; if $CO(p) = \{y\}$, then $c_p(y,x) > c_p(x,y)$. Note that the only $B \subseteq A$ such that $B \succ_{\mathbf{K}(q)} \{y\}$ is $B = \{x\}$. Since $\mathrm{rank}_q(x) = 1$, we have $c_{p'}(y,x) > c_{p'}(x,y)$, thus $CO(p') \neq \{x\}$. If $CO(p) = \{z\}$, then $c_p(z,x) > c_p(x,z)$ and $c_p(z,y) > c_p(y,z)$. Since $\mathrm{rank}_q(z) = 3$, we have $c_{p'}(z,x) > c_{p'}(x,z)$ and $c_{p'}(z,y) > c_{p'}(y,z)$. Thus $CO(p') = \{z\}$. In any case, it is not possible to have $CO(p') \succ_{\mathbf{K}(q)} CO(p)$.
        \item If $|CO(p)| = 3$, we have  $CO(p) = \{x,y,z\}$ and thus the only $B \subseteq A$ such that $B \succ_{\mathbf{K}(q)} \{y\}$ is $B = \{x\}$. Since $\mathrm{co}(x,p) = \mathrm{co}(y,p) = \mathrm{co}(z,p) = 0$, we have $c_p(x,y) < c_p(y,x)$ or $c_p(x,z) < c_p(z,x)$. Now,  $CO(p') = \{x\}$ holds if and only if $c_{p'}(x,y) > c_{p'}(y,x)$ and $c_{p'}(x,z) > c_{p'}(z,x)$. But since $\mathrm{rank}_q(x) = 1$, that is not possible. As a consequence we do not have $CO(p') \succ_{\mathbf{K}(q)} CO(p)$. 
    \end{itemize}
    We conclude that $CO$ is $\mathbf{K}$-strategy-proof.
\end{proof}

\begin{proposition}\label{39}
    Assume that $n = 4$ and $m$ is odd. Then, $CO$ is $\Omega^{CO}$-$\mathbf{K}$-manipulable.
\end{proposition}

\begin{proof}
    Let $q \coloneq [1,2,3,4]$ and $\omega \coloneq \{\overline{p} \in \mathcal{L}(A)^{I \setminus \{m\}}: CO(\overline{p},q[m]) = \{3,4\}\}$. Consider $\overline{p} \in \mathcal{L}(A)^{I \setminus \{m\}}$ defined as follows: 
    \begin{itemize}
        \item for every $j \in \ldbrack\frac{m-1}{2}\rdbrack$, $\overline{p}(j) \coloneq [4,2,3,1]$;
        \item for every $j \in \{\frac{m-1}{2},\dots,m-1\}$, $\overline{p}(j) \coloneq [3,4,1,2]$.
    \end{itemize}
    We have 
   $\mathrm{co}(1,(\overline{p}, q[m])) = -1;\ \mathrm{co}(2,(\overline{p}, q[m])) = -1;\ \mathrm{co}(3,(\overline{p}, q[m])) = 1;\
\mathrm{co}(4,(\overline{p}, q[m])) = 1.$
    Thus, $CO(\overline{p}, q[m]) = \{3,4\}$. Hence, $\omega \neq \varnothing$ and we have $\omega \in \Omega^{CO}_i(q)$. Let now $\psi$ be the permutation that exchanges $2$ and $3$ and $q'=\psi q$. We have 
    $\mathrm{co}(1,(\overline{p}, q'[m])) = -1;\ \mathrm{co}(2,(\overline{p}, q'[m])) = -3;\ \mathrm{co}(3,(\overline{p}, q'[m])) = 3;\ \mathrm{co}(4,(\overline{p}, q'[m])) = 1.$
    Thus $CO(\overline{p}, q'[m]) = \{3\} \succ_{\mathbf{K}(q)} \{3,4\} = CO(\overline{p}, q[m])$. Take now any $\overline{p}' \in \omega$. By definition we have $CO(\overline{p}, q[m]) = \{3,4\}$. Thus $\mathrm{co}(3,(\overline{p}', q[m])) = \mathrm{co}(4,(\overline{p}', q[m])) > \mathrm{co}(1,(\overline{p}', q[m]))$ and $\mathrm{co}(3,(\overline{p}', q[m])) > \mathrm{co}(2,(\overline{p}', q[m]))$. Consider now $(\overline{p}',q'[m])$. If $c_{(\overline{p}', q[m])}(3,2) > c_{(\overline{p}', q[m])}(2,3)$ or $c_{(\overline{p}', q[m])}(2,3) > c_{(\overline{p}', q[m])}(3,2) + 1$, then, for every $x \in A$, $\mathrm{co}(x,(\overline{p}', q[m])) = \mathrm{co}(x,(\overline{p}', q'[m]))$, thus $CO((\overline{p}', q'[m])) = CO(\overline{p}', q[m])$. If otherwise $c_{(\overline{p}', q[m])}(2,3) = c_{(\overline{p}', q[m])}(3,2) + 1$, then we have 
    $\mathrm{co}(1,(\overline{p}', q'[m])) = \mathrm{co}(1,(\overline{p}', q[m]));\ 
        \mathrm{co}(2,(\overline{p}', q'[m])) = \mathrm{co}(2,(\overline{p}', q[m])) - 2;\ 
        \mathrm{co}(3,(\overline{p}', q'[m])) = \mathrm{co}(3,(\overline{p}', q[m])) + 2;\ 
        \mathrm{co}(4,(\overline{p}', q'[m])) = \mathrm{co}(4,(\overline{p}', q[m])).$
    Thus $CO((\overline{p}', q'[m])) = \{3\}$. In any case, we have $CO((\overline{p}', q'[m])) \succeq_{\mathbf{K}(q)} CO((\overline{p}', q[m]))$. Then $CO$ is $\Omega^{CO}$-$\mathbf{K}$-manipulable.
\end{proof}
We are now ready to prove Theorem \ref{Copeland-main}.
\begin{proof}[Proof of Theorem \ref{Copeland-main}] Consider Propositions \ref{34}, \ref{35}, \ref{36}, and \ref{40}, and apply Theorem \ref{general}, and Propositions \ref{37}, \ref{38}, and \ref{39}.
\end{proof}

\section*{References}

\noindent Andjiga, N.G., Mbih, B., Moyouwou, I., 2008. Manipulation of voting schemes with restricted beliefs
Journal of Mathematical Economics 44, 1232-1242.
\vspace{2mm}

\noindent Bandyopadhyay, T., 1982. Threats, counter-threats and strategic manipulation for non-binary group decision rules. Mathematical Social Sciences 2, 145-155.
\vspace{2mm}

\noindent Bandyopadhyay, T., 1983. Manipulation of non-imposed, non-oligarchic, non-binary group decision rules. Economics Letters 11, 69-73.
\vspace{2mm}

\noindent Barberà, S., 1977a. Manipulation of social decision functions. Journal of Economic Theory 15, 266-278.
\vspace{2mm}

\noindent  Barberà, S., 1977b. The manipulation of social choice mechanisms that do not leave “too much” to chance. Econometrica 45, 1573-1588.
\vspace{2mm}

\noindent Barberà, S., Dutta, B., Sen, A., 2001. Strategy-proof social choice correspondences. Journal of Economic Theory 101, 374-394.
\vspace{2mm}

\noindent Barberà, S., Bossert, W., Pattanaik, P.K., 2004. Ranking sets of objects. In: Barberà, S., Hammond, P.J., Seidl, C. (eds) Handbook of Utility Theory. Springer, Boston, MA.
\vspace{2mm}

\noindent Barberà, S., 2011. Strategy-proof social choice. Handbook of Social Choice and Welfare, volume 2, 731-831.
\vspace{2mm}

\noindent Bebchuk, L.A., 1980. Ignorance and manipulation. Economics Letters 5, 119–123.
\vspace{2mm}

\noindent Brandt, F., Brill, M., 2011. Necessary and sufficient conditions for the strategy-proofness of irresolute social choice functions. In Proceedings of the 13th Conference on Theoretical Aspects of Rationality and Knowledge (TARK XIII). Association for Computing Machinery, New York, NY, USA, 136-142.
\vspace{2mm}

\noindent Brandt, F., 2015. Set-monotonicity implies Kelly-strategyproofness. Social Choice and Welfare 45, 793-804.
\vspace{2mm}

\noindent Brandt, F., Bullinger, M., Lederer, P., 2022a. On the indecisiveness of Kelly-strategyproof social choice functions. Journal of Artificial Intelligence Research, 1093-1130.
\vspace{2mm}

\noindent Brandt F., Saile, C., Stricker C., 2022b. Strategyproof social choice when preferences and outcomes may contain ties. Journal of Economic Theory 202, 105447.
\vspace{2mm}

\noindent Bubboloni, D., Gori, M., 2016. Resolute refinements of social choice correspondences. Mathematical Social Sciences 84, 37-49.
\vspace{2mm}

\noindent Bubboloni, D., Gori, M., 2021. Breaking ties in collective decision-making. Decisions in Economics and Finance 44, 411-457.
\vspace{2mm}

\noindent Campbell, D.E., Kelly, J.S., Qi, S., 2018. A stability property in social choice theory. International Journal of Economic Theory 14, 85-95.
\vspace{2mm}

\noindent Ching, S., Zhou, L., 2002. Multi-valued strategy-proof social choice rules. Social Choice and Welfare 19, 569-580.
\vspace{2mm}

\noindent Conitzer, V., Walsh, T., Xia, L., 2011. Dominating manipulations in voting with partial information. In: Proceedings
of the Twenty-Fifth AAAI Conference on Artificial Intelligence, 638-643.
\vspace{2mm}

\noindent Copeland, A.H., 1951. A reasonable social welfare function. In: University of Michigan
Seminar on Applications of Mathematics to the social sciences.
\vspace{2mm}

\noindent Duggan, J., Schwartz, T., 2000. Strategic manipulability without resoluteness
or shared beliefs: Gibbard-Satterthwaite generalized. Social Choice and Welfare 17, 85-93.
\vspace{2mm}

\noindent Endriss, U., Obraztsova, S., Polukarov, M., Rosenschein, J.S., 2016. Strategic voting with incomplete information. In: Proceedings of the Twenty-Fifth International Joint Conference on Artificial Intelligence. IJCAI-16, pp. 236-242.
\vspace{2mm}

\noindent Farquharson, R., 1969. Theory of Voting. Yale University Press, New Haven.
\vspace{2mm}

\noindent Fishburn, P.C., 1977. Condorcet social choice functions. SIAM J. Appl. Math. 33,
469–489.
\vspace{2mm}

\noindent G\"ardenfors, P., 1976. Manipulation of social choice functions. Journal of Economic Theory 13, 217-228.
\vspace{2mm}

\noindent Gibbard, A., 1973. Manipulation of voting schemes: A general result. Econometrica 41, 587-601.
\vspace{2mm}

\noindent Gori, M., 2021. Manipulation of social choice functions under incomplete information. Games and Economic Behavior 129, 350-369.
\vspace{2mm}

\noindent Kelly, J.S., 1977. Strategy-proofness and social choice functions without singlevaluedness. Econometrica 45, 439-446.
\vspace{2mm}

\noindent MacIntyre, I., Pattanaik, P.K., 1981. Strategic voting under minimally binary group decision functions. Journal of Economic Theory 25, 338-352. 
\vspace{2mm}

\noindent Moulin, H., 1981. Prudence versus sophistication in voting strategy. Journal of Economic Theory 24, 398-412.
\vspace{2mm}

\noindent Nurmi, H., 1987. Comparing Voting Systems. D. Reidel Publishing Company, Dordrecht.
\vspace{2mm}

\noindent Nehring, K., 2000. Monotonicity implies generalized strategy-proofness for correspondences. Social Choice and Welfare 17, 367-375.
\vspace{2mm}

\noindent Pattanaik, P.K., 1975. Strategic voting without collusion under binary and democratic group decision rules. The Review of Economic Studies 42, 93-103.
\vspace{2mm}

\noindent     Sanver, M.R.,  Zwicker, W.S., 2012. Monotonicity properties and their adaptation to irresolute social choice rules. Social Choice and Welfare 39, No. 2/3, Special Issue in Honour of Maurice Salles on Developments in Social Choice and Welfare Theories, 371--398.
\vspace{2mm}

\noindent Reijngoud, A., Endriss, U., 2012. Voter response to iterated poll information. In: Proceedings of the 11th International Conference on Autonomous Agents and Multiagent Systems. AAMAS, pp. 635-644.
\vspace{2mm}

\noindent Sato, S., 2008. On strategy-proof social choice correspondences. Social Choice and Welfare 31, 331-343.
\vspace{2mm}

\noindent Satterthwaite, M.A., 1975. Strategy-proofness and Arrow's conditions: Existence and correspondence theorems for voting procedures and social welfare functions. Journal of Economic Theory 10, 187-217. 
\vspace{2mm}

\noindent Sengupta, M., 1978. On a difficulty in the analysis of strategic voting. Econometrica 46, 331-343.
\vspace{2mm}

\noindent Sengupta, M., 1980. The knowledge assumption in the theory of strategic voting. Econometrica 48, 1301-1304 (Notes and comments).
\vspace{2mm}

\noindent Taylor, A.D., 2002. The manipulability of voting systems. The American Mathematical Monthly 109, 321-337.
\vspace{2mm}

\noindent Taylor, A.D., 2005. Social choice and the mathematics of manipulation. Cambridge: Cambridge University Press.
\vspace{2mm}

\noindent Terzopoulou, Z., Endriss, U., 2019. Strategyproof judgment aggregation under partial information. Social Choice and Welfare 53, 415-442.
\vspace{2mm}

\noindent Tsiaxiras, J., 2021. Strategic voting under incomplete information in approval-based committee elections. MSc Thesis. University of Amsterdam. Available at \url{https://eprints.illc.uva.nl/id/eprint/1799/1/MoL-2021-09.text.pdf}
\vspace{2mm}

\noindent Veselova, Y.A., 2020. Does incomplete information reduce manipulability?. Group Decision and Negotiation 29, 523-548. 
\vspace{2mm}

\noindent Veselova, Y.A.,  Karabekyan, D., 2023. Simultaneous manipulation under incomplete information. Working paper. Available at \url{https://api.semanticscholar.org/CorpusID:262231173}
\vspace{20mm}

\end{document}